\documentclass[%
nofootinbib,        
amsmath,amssymb,
aps,
floatfix,
]{revtex4-2}

\usepackage{amsthm}                     
\usepackage[inline]{enumitem}           

\usepackage{tikz}                       
\tikzset{                               
node/.style={on grid},                  
}                                       
\usetikzlibrary{arrows}                 

\newcommand{\mnote}[1]{}                                        

\newtheorem{lem}{Lemma}            
\theoremstyle{definition}
\newtheorem{defn}{Definition}
\theoremstyle{remark}
\newtheorem{rem}{Remark}

\begin{document}

\title[Numerical Expanding Direction of $T^2$-Symmetric Spacetimes]{A Numerical Study of the Expanding Direction of $T^2$-Symmetric Spacetimes}

\author{Beverly K. Berger}
\address{Kavli Institute for Particle Astrophysics and Cosmology,
Stanford University,
Stanford, CA 94305,
USA}
\email{beverlyberger@me.com}

\author{James Isenberg}
\address{Department of Mathematics,
University of Oregon,
Eugene, OR  97403-1222,
USA}
\email{isenberg@uoregon.edu}

\author{Adam Layne}
\email{adam@fnbu.pw}
\address{I. Šimulionio g. 5 - 106,
04330 Vilnius,
Lithuania}

\date{\today}

\begin{abstract}
    The asymptotic behavior of expanding, generic, $T^2$-Symmetric, vacuum spacetimes is examined via numerical simulations. After validation of the numerical methods, the properties of these generic spacetimes are explored and compared to non-generic subfamilies where proven results exist.
    The non-generic subfamilies within this class, including the Kasner, the Gowdy, the pseudo-homogeneous, and the $B=0$ spacetimes, all have known asymptotic behaviors in the expanding direction which have been determined either from the explicit solutions or using analytic methods. For the $B\ne 0$ spacetimes, the generic case within the $T^2$-Symmetric vacuum solutions, the asymptotic behavior has not been determined analytically. In this work, we use numerical simulations to explore the asymptotic behavior of the $B\ne 0$ spacetimes. Our results indicate that, for these generic spacetimes, the asymptotic behavior in the expanding direction differs from that seen in the non-generic subfamilies. 
   In addition to differences in asymptotic power laws, an apparent quasi-periodic exchange of energy from one gravitational mode to the other for the generic non-polarized solutions is observed. 
\end{abstract}

\maketitle

\section{introduction}

Although very little is known about the behavior of general solutions of Einstein's equations -- with or without matter sources -- studies of families of solutions of these equations characterized by the presence of nontrivial isometries have been carried out which could provide some insight into the behavior of solutions which do not have any symmetries.
The Friedman--Lema\^{i}tre--Robinson--Walker (FLRW) solutions, which are spatially homogeneous, isotropic, and have fluid source fields, are well understood. In recent work \cite{MR3874696}, it has been shown that, for initial data sufficiently close to FLRW initial data (with no isometries presumed), the behavior near the singularity of solutions developed from this data is much like that of the FLRW solutions.
Numerical studies \cite{MR1858721,PhysRevD.79.123526,GARFINKLE2020}
have also indicated that the ``local Mixmaster'' behavior near the singularity of vacuum solutions of the Einstein equations with \(SU(2)\) spatial isometry is observed in a wider class of solutions.\footnote{There is some numerical evidence \cite{PhysRevD.79.123526} that local Mixmaster dynamics does not describe the complete behavior of these spacetimes in the collapsing direction.}
One particular family of solutions with nontrivial isometries for which a significant amount has been learned are those which are labeled as the ``\(T^2\)-Symmetric'' spacetimes.
These spacetimes all have Cauchy surfaces with \(T^3\) topology, and have the isometry group \(T^2\) acting on and preserving these Cauchy surfaces.
It has been proven that the \(T^2\)-Symmetric vacuum spacetime solutions all admit global areal foliations \cite{MR1474313} as well as constant mean curvature (CMC) foliations \cite{Andreasson:2002ke} which distinguish the expanding and contracting directions. We emphasize that the spacetimes in this class do not describe the physical universe in any way, but the study of this class should serve as a theoretical laboratory to gain insight into the behavior of generic solutions of Einstein’s equations.

While most studies of \(T^2\)-Symmetric spacetimes have focused on the behavior in the contracting direction -- near the singularity -- there has recently been progress in understanding how these solutions behave in the expanding direction.
The work of LeFloch and Smulevici \cite{MR3513138} as well as that of the authors of this paper \cite{MR4062458} shows that, for certain restricted subfamilies of the \(T^2\)-Symmetric spacetimes, there are local attractors (in a certain sense which we discuss below) which characterize the evolution of these vacuum spacetimes in the expanding direction.
The question then arises if there are local attractors in the expanding direction for the general class of \(T^2\)-Symmetric spacetimes.
More importantly, one might speculate if the set of all vacuum solutions of the Einstein equations -- perhaps restricted to those which have \(T^3\) Cauchy surfaces, and therefore have distinguished expanding and contracting directions -- have local attractors.

We do not address here this last question concerning the behavior of general vacuum solutions.
Rather we focus on establishing evidence for the existence of local attractors for the general \(T^2\)-Symmetric spacetimes. Note that in this paper, we use ``attractor'' in a sense similar to that used in \cite{MR4062458} such that the asymptotic dynamics of the appropriate spatially averaged variables approaches a ``sink''. Such behavior is shown in some detail below. In addition, the asymptotic values of the main variables are dominated by their spatial averages (functions of the time variable $\tau$ only) as is shown in Table I.
Unlike the theorems in \cite{MR3513138,MR4062458} which prove that the polarized and the \(B=0\) subfamilies of the \(T^2\)-Symmetric spacetimes (we define these and other subfamilies of solutions below) have local attractors in the expanding direction, the evidence which we present here for the existence of local attractors in the unrestricted family of  \(T^2\)-Symmetric spacetimes is based on numerical simulations of the evolution of these solutions.
While such numerical studies are not as deterministic as mathematical theorems, we note that it was a set of numerical simulations which led us to conjecture the behavior which we have proven \cite{MR4062458} holds for the \(B=0\) subfamilies of these spacetimes.
We also note that the numerical investigations discussed in this paper indicate further interesting behavior within these spacetimes.
In particular, the numerical simulations indicate that, presuming both of the gravitational modes are activated, then the ``energy'' periodically sloshes equally from one mode to the other, in the manner of the ``equipartition behavior'' observed in coupled harmonic oscillator systems.
We discuss this property in Section \ref{sloshing}.

\subsection{Coordinates and the Einstein Equations for $T^2$-Symmetric Spacetimes}

The \(3+1\) dimensional \(T^2\)-Symmetric spacetimes under consideration all admit a global foliation by spacelike hypersurfaces in terms of which the metric takes the following form \cite{MR1474313}:
\begin{align}
    g
    =	&	e^{\widehat l - V + 4\tau} \left( - d\tau^2 + e^{2(\rho - \tau)} d\theta^2 \right) 	+e^{V} \left[ dx + Q dy 
    + (G + Q H ) d\theta \right]^2 	+  e^{-V+2\tau} \left[ dy + H d\theta \right]^2 \,.				\label{coords}	
\end{align}
In this formula, the metric coefficients \(V, Q, \rho, \widehat l, G, H\) are all functions of the time coordinate \(\tau\) and the (periodic) spatial coordinate \(\theta\).
The vector fields \(\partial_x\) and \(\partial_y\) are the Killing vector fields which generate the \(T^2\) isometry.
The areas of the orbits of these vector fields are a constant multiple of \(e^{2\tau}\), so that the spacetime singularity occurs as \(\tau \to - \infty\), and the spacetime expands as \(\tau\) becomes increasingly positive.
In \cite{MR4062458}, we describe the relationship between the coordinates and variables that we use here and those used in other studies of \(T^2\)-Symmetric spacetimes.

In terms of the metric coefficients which we use here, the Einstein field equations take the form
\begin{align}
    \partial_\tau \left( e^\rho V_\tau \right)
    =& \partial_\theta \left( e^{2\tau - \rho} V_\theta \right)+ e^{2(V-\tau)+\rho} \left( Q_\tau^2 - e^{2(\tau-\rho)} Q_\theta^2 \right) \label{v transport}	\ ,				\\
    \partial_\tau \left( e^{\rho + 2 (V - \tau)} Q_\tau \right)  \label{q transport}
    =& \partial_\theta \left( e^{-\rho +2 V} Q_\theta \right) \ ,				\\
    \widehat l_\tau + \rho_\tau + 2 
    =&\frac{1}{2} \left[ V_\tau^2 + e^{2(\tau-\rho)} V_\theta^2 + e^{2(V-\tau)} \left( Q_\tau^2 + e^{2(\tau-\rho)}  Q_\theta^2 \right) \right]\label{l-evol} 		\ ,			\\
    \rho_\tau
    =&  \frac{K^2}{2}e^{\widehat l}\label{rho-evol}	\ ,				\\
    \widehat l_\theta=&V_\theta V_\tau + e^{2(V-\tau)} Q_\theta Q_\tau \label{constraint} \ ,	
    \\
    G_\tau =&-Q \, e^{\rho + \widehat l} \, K \label{Geq}\ ,
    \\
    H_\tau =& e^{\rho + \widehat l} \, K \label{Heq} \ ,
\end{align}
where we use the subscripts \(\tau,\, \theta\) to indicate partial derivatives.  The twist, $K$, is defined below in Eqs. \eqref{x twist} and \eqref{y twist}. For vanishing twist, Eqs. \eqref{v transport} and \eqref{q transport} are wave equations for the two gravitational wave polarizations here written in terms of $V$ and $Q$. In this case, the remaining variables may be constructed from any given solution $\{V,Q \}$. Eqs \eqref{l-evol} and \eqref{constraint} are respectively the Hamiltonian and $\theta$-momentum constraint while Eq. \eqref{rho-evol} vanishes identically if $K = 0$.
In these equations, the quantity \(K\) represents the ``twist constant'' of the Killing vector fields.
For general \(T^2\)-Symmetric spacetimes, in terms of the coordinate systems which we use here, there are generally two non-vanishing twist constants which take the following form:
\begin{align}
    K(\partial_x)	=	&	- e^{\widehat l - \rho+2(V-\tau)}   \left( G_\tau + Q H_\tau\right) ,	\label{x twist}\\
    K(\partial_y)	=	&	- e^{\widehat l - \rho} \left[ H_\tau + e^{2(V-\tau)} Q \left( G_\tau + Q H_\tau\right) \right].\label{y twist}
\end{align}
We note, however, that, without loss of generality, we can always adjust these coordinates so that \(K(\partial_x)\) vanishes (see footnote 1 in \cite{MR4062458}).
For the family of solutions we consider here, the remaining twist constant \(K(\partial_y)\) in Eq. \eqref{y twist} does not vanish, and we label it \(K\). The remaining metric variables $G$ and $H$ are non-dynamical in the sense that, although they  satisfy equations, none of the previously listed Einstein field equations depend on their values. For a complete discussion, see \cite{MR1858721} where they are called $G_1$ and $G_2$ respectively.
For the subfamily of Gowdy spacetimes, both twist constants vanish.
Our focus in this paper is on the non-Gowdy \(T^2\)-Symmetric spacetimes, so we presume that \(K\ne0\).
For such spacetimes, it is useful to make the following substitution:
\begin{align}
    l := \widehat l + \ln (K^2/2) \label{lsub}.
\end{align}
In our discussion of the non-Gowdy, \(T^2\)-Symmetric Einstein evolutions below, we make this substitution and refer to \(l\) rather than to \(\widehat l\) although, for clarity, we continue to use $\widehat l$ in referring to equations prior to Eq. \eqref{lsub}. Starting with Eq. \eqref{a def}, $l$ is used in place of $\widehat l$.

Equation \eqref{l-evol} is actually a consequence of the Einstein constraint equations, but we use it as an evolution equation here; see Section 1 of \cite{MR4062458} for additional discussion.\footnote{
This also occurs, for example, in isotropic models where the evolution of the spatial metric scale factor is determined by the Hamiltonian constraint.}
Equations \eqref{v transport} through \eqref{rho-evol} specify the evolution for \(V, Q,\rho\), and \(\widehat l\).
Equation \eqref{constraint}	 is a constraint, viz. the momentum constraint, on the initial data of a given solution.
Equations \eqref{x twist} and \eqref{y twist} constitute first order evolution equations for \(G\) and \(H\), and are semi-decoupled from Eqs. \eqref{v transport}-\eqref{constraint} in the sense that Eqs. \eqref{v transport}-\eqref{constraint} do not involve \(G\) and \(H\); consequently, the evolution of these two quantities can be carried out after the solution of Eqs. \eqref{v transport}-\eqref{constraint} has been determined.

To define the subfamilies of the \(T^2\)-Symmetric spacetimes for which a number of strong results have been proven, it is useful to note that there are two quantities constructed as spatial averages of the field variables which are conserved by the Einstein evolution equations,  not only in all the subfamilies, but also in the generic case:
\begin{align}
    A:=&\int_{S^1}  e^\rho \left( V_\tau  - e^{2(V-\tau)} Q_\tau Q \right) \, d\theta,	\label{a def}		\\
    B:=&\int_{S^1}  e^{\rho+2(V-\tau)} Q_\tau  \, d\theta.			\label{b def}
\end{align}
The value of the conserved quantity $A$ does not apparently play a major role in distinguishing subfamilies for which results regarding the asymptotic behavior (in the expanding direction) have been determined either analytically or from numerical simulations. In contrast, the value of the conserved quantity $B$ -- whether it is zero or not -- is important in this regard. Based on the vanishing or not of $B$ along with other characterizations of the $T^2$-Symmetric solutions listed below in Definition \ref{subfamilies}, one can define $12$ subfamilies of the full family, including the generic class. These are schematically diagrammed in Fig. \ref{venn}. Note that, as a consequence of results proven in Section \ref{familyreduce}, two of the expected subfamilies do not exist, so only $10$ subfamilies are indicated in Fig. \ref{venn}.
\begin{defn}\label{subfamilies}
    Let \((\widehat l, V, Q , \rho)\) specify a solution of Eqs. \eqref{v transport}-\eqref{constraint}, and let \(g\) be the metric corresponding to that solution.
    \begin{itemize}
        \item
        \(g\) is \emph{pseudo-homogeneous} (PH) if the metric coefficients \(V, Q\) and \(\widehat l\) are spatially constant \cite{MR3312437}.
        \footnote{
        The definition of pseudo-homogeneous in \cite{MR3312437} includes the condition that the solutions are non-Gowdy.
        We use a more general definition here.
        }
        \item
        \(g\) is \emph{polarised} if \(Q \equiv 0\).
        \footnote{It is more usual to define polarized spacetimes using a coordinate free, geometric condition. For the spacetimes under consideration, this geometric condition is equivalent to the one we give here.
        See footnote 1 of \cite{MR4062458}.}
        \item
        \(g\) is \emph{Gowdy} if \(H,G\), and \(\rho\) all vanish identically (and thus \(K = 0\)).
          \end{itemize}
       {\it Remark}. {We are abusing definitions slightly here.
        The definition given here for the Gowdy solutions \cite{GOWDY1974203} corresponds to the subset of Gowdy solutions considered in \cite{MR3312437}.
        Such solutions do not encompass the entire Gowdy class.
        More generically, \emph{Gowdy} solutions are those for which \(K = 0\) and \(\rho \equiv 0 \).
        Together with the formulas for the twists \eqref{x twist} and \eqref{y twist}, this implies that \(G_\tau = H_\tau = 0\), although $G$ and $H$ may generically have spatial variation.
        Since the equations for \(G,H\) are semi-decoupled from the evolution equations of the remaining functions, the class of solutions considered in \cite{MR3312437} still encompasses all the possible behavior of \(\widehat l, V,Q\) that could occur in this larger class.}
  
\end{defn}

\begin{figure}
    \centering
    \includegraphics[width=0.4\textwidth]{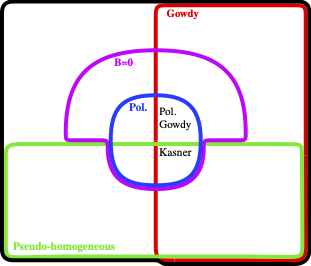}
    \caption{The subfamilies of \(T^2\)-Symmetric vacuum spacetimes on \(T^3\). This diagram represents the $10$ subfamilies of the $T^2$-symmetric solutions. The subfamilies are defined by intersections of the characteristics listed in Definition \ref{subfamilies}, along with the vanishing or not of the conserved quantity $B$. There are $10$ rather than $12$ subfamilies as a consequence of the results proven in Section \ref{familyreduce}. The generic subfamily is in the upper left corner. We note that the figure is schematic in the sense that the shapes and the relative scales of the regions corresponding to the subfamilies are arbitrary and selected for clarity.}\label{venn}
\end{figure}
\mnote{JIM: Are we only discussing in figure 1 asymptotic behavior for solutions which is known from mathematical results, or are we also indicating in figure 1 the asymptotic behavior which we have learned from these numerical simulations? We should be clear about this.}

Our focus in this paper is on numerical simulations of the evolutions (in the expanding direction) of $T^2$-Symmetric solutions which are not contained in the nine special subfamilies indicated in Fig. \ref{venn}; i.e., on the solutions in the upper left corner of this figure.
However, to best understand the behavior of the general \(T^2\)-Symmetric spacetimes, it is useful to contrast the behavior of some of the nine special subfamilies with that seen in the general  \(T^2\)-Symmetric spacetimes.
We do this throughout the paper.
Since numerical simulations of certain nonlinear PDE systems play a central role in this work, it is important that we describe the procedure that we have used for carrying out these numerical studies.
We do this in section \ref{numerical}.
Also in that section, we discuss some of the features of our numerical work which we believe lead us to expect that the behavior seen in the numerical simulations occurs generally in $T^2$-Symmetric solutions. Currently, this belief relies on numerical results from many simulations arising from what appears to be a numerical approximation of an open set of initial data sets.
Ideally, we would like to prove that the observed behavior in the numerical simulations does occur generally.
We have not yet been able to do this.
However, in Section \ref{impliedasymptotic}.B, we outline a strategy that we hope will lead to a theorem confirming that the behavior we observe numerically is the generic behavior for \(T^2\)-Symmetric spacetimes in the expanding direction. Whether or not this strategy does lead to our proving mathematical results describing the asymptotic behavior in the expanding direction of $T^2$-Symmetric vacuum solutions, we believe that these numerical simulations provide a very useful guide for what behavior to expect to verify for these spacetimes.

\subsection{Questions we address in this paper.}
As we discuss in Section \ref{math results}, certain behaviors in the expanding direction have been mathematically verified for the non-generic subfamilies of the $T^2$-Symmetric solutions. In particular, as we detail in Section \ref{math results}, it has been proven that the spatial integrals of certain of the metric components in these subfamilies of solutions have prescribed asymptotic behavior in the expanding direction, characterized by attractors. One is then led to question whether these subfamily attractors continue to be attractors for the generic $T^2$-Symmetric solutions. The numerical results which we present in this paper indicate that this is not the case; we obtain numerical evidence that the generic $T^2$-Symmetric solutions approach attractors in the expanding direction which are different from those for the non-generic subfamilies of solutions.
\mnote{?: actually enumerate what is proved and partially proved -- BKB cannot do this}

In this paper and our previous paper \cite{MR4062458}, we consider only $T^2$-Symmetric vacuum spacetimes with non-zero twist constant (see \cite{MR1474313} for a more general discussion of these spacetimes) in the expanding direction. The zero twist case (Gowdy spacetimes) was first considered by one of us for polarized models \cite{BERGER1974458} and later by many others. The expanding direction has been explored in detail by Ringstr\"om \cite{MR2217283}. He uses the spatially integrated constants $A$,$B$, and $C$
to identify a type  of expanding behavior in the Gowdy spacetime. 
We note that $A$, $B$ are the conserved quantities defined above in equations \eqref{a def} and \eqref{b def}, while $C$, defined as
\begin{align}
    C:= \int_{S^1} \left[ \,Q_\tau \left( 1 - e^{2(V-\tau)} Q^2 \right) + 2 Q (V_\tau -1)\,\right] \, d\theta\ , \label{c def}
\end{align}
is a quantity which is conserved in Gowdy spacetimes but not in $T^2$-Symmetric solutions with non-zero twist.
$A$, $B$, and $C$ can be combined to yield 
\begin{align}
    \zeta := A^2 +BC. \label{zeta}
\end{align}

 In \cite{MR3312437}, Ringstr\"om shows that for the evolution in the expanding direction of non-Gowdy vacuum $T^2$-Symmetric spacetimes restricted to what we now know as the $B=0$ subclass, there appears to be an attractor (which has also been found in numerical simulations).
The polarized case with non-zero twist has been shown in \cite{MR3513138} to have the asymptotics indicated in Table I for the $B=0$ subclass (but with the restriction to small data). In \cite{MR4062458} we show that the $B = 0$ subclass (not necessarily polarized) has the asymptotics listed for that class in Table I. Here we provide the numerical analysis of the generic case with $B \ne 0$.  The asymptotic behavior listed in Table I describes the ``attractor'', as we use the term in \cite{MR4062458}, for expanding $T^2$-Symmetric spacetimes. While the behavior of the attractor for the $B = 0$ class is proven in \cite{MR4062458}, that for the generic case is only known numerically.

\section{Mathematical Results for Non-generic Subfamilies}\label{math results}
\mnote{JIM: draft rewrite is below at Section \ref{jim draft 1}}
Before we present our numerical simulations and the implications we infer from them regarding the behavior of generic $T^2$-Symmetric solutions in the expanding direction, we summarize in this section what has been determined mathematically for the solutions contained in the nine special subfamilies of solutions. We restrict our attention to the mathematically proven behavior in the expanding direction. For those subfamilies for which explicit solutions have been found, we determine the asymptotic behavior of those subfamilies directly from these explicit solutions.

\subsection{Reduction of the Number of Subfamilies of $T^2$-Symmetric Vacuum Solutions}\label{familyreduce}

In this subsection, we state and prove a lemma which eliminates two of the possible subfamilies of $T^2$-Symmetric solutions indicated by the conditions in Definition \ref{subfamilies}, so that only nine special subfamilies remain, as discussed above.

\begin{lem}
    Solutions which are both PH and \(B=0\) are polarized.
\end{lem}

\begin{proof}
    The PH and B=0 conditions together imply
    \begin{align}
        0=B=\int_{S^1} e^{\rho + 2(V - \tau) } Q_\tau \, d\theta = e^{ 2(V - \tau) } Q_\tau \int_{S^1} e^\rho\, d\theta. \label{PH lemma}
    \end{align}
    Note that the first and last factors on the right side of equation \eqref{PH lemma} are strictly positive.
    It follows that  \(Q_\tau = 0\).
    In this case we may make a linear substitution \(d \widetilde x = dx + Q \, dy , d\widetilde y = dy\) to obtain a solution which satisfies \(Q \equiv 0\).
\end{proof}

\subsection{Kasner solutions}

For the purposes of this paper, we describe only the relevant properties of the Kasner solutions, i.e., those properties related to the behavior of these solutions in the expanding direction. The metric \eqref{coords} yields the Kasner spacetime if spatial homogeneity is assumed (no dependence on $\theta$). This assumption allows the choices $\rho =$ constant, $Q=0, \ G = 0, \ H=0$ leading to the solution (where $A$ is defined in \eqref{a def})
\begin{align}
    V = A \tau + b , \quad \widehat l = \left(\frac{1}{2}A^2 - 2 \right) \tau + c
\end{align}
for some constants \(A,b,c \in \mathbb R\) \cite{MR1501305}.

\subsection{Polarized Gowdy solutions}

The next level of complexity is to allow dependence on $\theta$ as defined in \eqref{coords}. These models have been first discussed by Gowdy \cite{GOWDY1974203} for all allowed spatial topologies ($T^3$, $S^3$, and $S^2 \times S^1$). In this paper, we restrict our discussion to the solutions with the $T^3$ topology. For the polarized case, there is a family of  explicit solutions in terms of zero order Bessel functions (see \cite{BERGER1974458}) satisfying the time-dependent-frequency wave equation obtained from \eqref{v transport} specialized to the polarized Gowdy case (zero $Q$ and constant $\rho$). Most interestingly, we note  for these explicit solutions, while $V$ has the same asymptotic behavior as the explicit Kasner solutions, $V_\tau$ is dominated as $\tau \to \infty$ by the wave energy term which grows as $ e^{\tau/2}$  rather than $\tau$ so that $\widehat l$ grows exponentially in $\tau$. Ringstr{\"o}m has called this qualitative change in the dynamics an instability \cite{MR3186493}. 

In the polarized case, 
\(V\) satisfies a linear wave equation with an explicit solution in terms of zero-order Bessel functions and thus inherits their known asymptotics.
These asymptotics can then be used to obtain the asymptotics of $\widehat l$ via Eq. \eqref{l-evol}.
Note that, in addition to the wave modes, there is a zero mode Kasner component that survives if all the coefficients of the spatially dependent terms vanish. 

\subsection{Nonpolarised, PH Gowdy models}
Solutions which are Gowdy and spatially homogeneous (this is equivalent to the PH condition in Gowdy spacetimes) but which are non-polarised have been considered in \cite{MR3312437}.
The solutions which are Gowdy, PH (and thus spatially homogeneous), and satisfy the condition \(B = 0\) (and so are polarized) are the \(T^3\) Kasner solutions mentioned above.

\subsection{Generic Gowdy spacetimes \label{generic Gowdy}}

Generic Gowdy solutions add the second gravitational-wave polarization $Q$ to the equations, which now take the form of two-degrees-of-freedom wave equations with nonlinear couplings. The constants $A$, $B$, and $C$ from Eqs. \eqref{a def}, \eqref{b def}, and \eqref{c def} define the dynamics of this case. The combination $\zeta = (A^2 + BC)$ has been shown by Ringstr{\"o}m \cite{MR2032917} to be especially relevant. Note that these constants involve spatial averages of nonlinear functions. This means that we can no longer assume that these averaged constants inherit the properties that the same combinations of nonlinear functions in the integrands would have in the homogeneous case. 

Ringstr{\"o}m shows \cite{MR2032917} that if $\zeta > 0$, the asymptotics of the model follows the polarized values. However, he has also determined that if $\zeta < 0$, there is a new class of solutions whose asymptotics are less clear. We point out that the definitions of $A$, $B$, and $C$ indicate that $\zeta$ must be $> 0$ if $B =0$ but {\it{may}} be $< 0$ if $B \ne 0$. (The actual sign of $\zeta$ depends on the magnitude of $BC$ compared to $A^2$). 


In this subsection, it is convenient to use $P:=V-\tau$ and $\lambda := 2\widehat l - 2 P + \frac{3}{2} \tau$ rather than $V$ and $l$ (as defined in Eq. \eqref{lsub}). In \cite{MR2032917}, Ringstr\"om describes the full asymptotics in the expanding direction of the Gowdy family of spacetimes. Note that this section focuses on the generic case of the Gowdy family that also includes the polarized Gowdy and PH Gowdy families. These spacetimes have a property that does not carry over to generic $T^2$-Symmetric spacetimes — namely, that the ``wave equations’’ for $P$ and $Q$ decouple from the ``background’’ metric variable $\lambda$. 

Given a solution of the wave equations for $P$ and $Q$, $\lambda$ can then be constructed therefrom. Note that $\lambda_\tau$ and $\lambda_\theta$ are determined (respectively) by the Hamiltonian constraint \eqref{l-evol} and the momentum constraint \eqref{constraint}. As mentioned above, the dynamics of this spacetime family can be expressed in terms of the integrated (over $S^1$) constants $A$, $B$, and $C$ where $A$, $B$, and $C$ are defined in \eqref{a def}, \eqref{b def}, and \eqref{c def} respectively. While $A$ and $B$ are also constants in the generic $T^2$-Symmetric family, $C$ is not. This means that there is no analog for the asymptotics of the Gowdy family in more generic families. Ringstr\"om has recognized that a particular combination of $A$, $B$, and $C$ — $\zeta$ from \eqref{zeta} — could have either sign. While the integrand is positive definite at any fixed value of $\theta$, the mean of a nonlinear function need not equal the function of the means of the constituents. Thus Ringstr\"om shows in \cite{MR2032917} that $P$ and $Q$ have power-law asymptotics if $\zeta > 0$ and oscillatory (bounded by $\pm$ a constant) asymptotics if $\zeta < 0$.\footnote{This oscillatory behavior is unrelated to the ``sloshing'' behavior in generic $T^2$-Symmetric spacetimes which we discuss below.}

For all values of $\zeta$, the asymptotic form of the spatial average of $\lambda$ is $\langle \lambda \rangle \to e^{{\tau}}$ \cite{MR2032917} -- the same as for the polarized Gowdy case. This means that the asymptotic behavior of the spatial average of $\lambda$ is linear in $\tau$ both for Kasner and for generic $T^2$-Symmetric spacetimes but exponential in $\tau$ for all (vacuum) Gowdy spacetimes. This behavior is a manifestation of instability as defined by Ringstrom \cite{MR3186493} that can be understood by examining the dominant term in the energy that provides a source for $\lambda$ \eqref{l-evol} as $\tau \to \infty$.

The much simpler case of polarized Gowdy asymptotics may be analyzed by setting $Q = 0$ for any of the above results.

There is, in fact, a homeomorphism between initial data on a chosen Cauchy surface and asymptotic data in the expanding direction \cite{MR1955788,MR2136117,2017arXiv170702803R} for the Gowdy spacetimes but this does not remain true for the generic $T^2$-Symmetric solutions.

\subsection{Polarized $T^2$-Symmetric spacetimes}

For the polarized (non-Gowdy) $T^2$-Symmetric spacetimes, there are no explicit solutions. However, in \cite{MR3513138}, LeFloch and Smulevici have mathematically proven that, for small initial data (in an appropriate sense, as detailed in \cite{MR3513138}), the asymptotic behavior in the expanding direction of the variables $V$, $\rho$, $l$ and their derivatives matches that indicated in Table \ref{tab}. Included in their results is the determination that spatial averages of some of these metric variables do evolve towards a distinct attractor, which is different from the attractor for the Gowdy spacetimes. This shows that the asymptotic behavior of the Gowdy spacetimes is not stable. The method used by LeFloch and Smulevici to prove their results involves careful determination of estimates for integrals of the metric coefficients, and the derivation of evolution equations for the integrals, with error terms playing a big role. The details of this analysis are presented in \cite{MR3513138}. We do note that this family of spacetimes necessarily has the quantity $B$ equal to zero.

\subsection{Generic $T^2$-Symmetric spacetimes with $B=0$}

If we remove the polarization condition $Q=0$, but continue to enforce the condition $B=0$, we get a subfamily of the $T^2$-Symmetric spacetimes which includes the polarized $T^2$-Symmetric solutions discussed above, but generalizes them significantly. As shown in \cite{MR4062458}, the method of analysis developed in \cite{MR3513138} generalizes to be applicable to these “$B_0$” solutions, and we have obtained a result very similar to that of \cite{MR3513138} for this much larger subfamily of spacetimes. In particular, we have shown in \cite{MR4062458} that for appropriately small initial data, the asymptotic behavior of the $B_0$ solutions is as described in Table \ref{tab}, and despite the presence of the nonzero $Q$ term, this asymptotic behavior is very similar to that seen in the polarized subfamily. Moreover, as in the polarized solutions, spatial averages of the metric components for the $B_0$ spacetimes do asymptotically approach an attractor.

A key step in generalizing the analysis of \cite{MR3513138} to the $B_0$ case is controlling the term $\int_{S^1} e^\rho V_\tau d\theta$. In the polarized case, this quantity is conserved during evolution. For the $B_0$ solutions, it is no longer conserved, but one verifies that the $B=0$ condition ensures that this quantity is bounded, consequently allowing an analysis similar to that in \cite{MR3513138} to be carried through in \cite{MR4062458}. Notably, if we consider solutions with $B\ne 0$, this boundedness no longer holds, which does not allow the analysis used in \cite{MR3513138} and \cite{MR4062458} to be used for the $B\ne 0$ case.

Interestingly, before we recognized that the $B_0$ solutions constitute a distinct subfamily of the $T^2$-Symmetric spacetimes, we did notice in our numerical simulations that there was anomalous asymptotic behavior. Further exploration showed us that one class of asymptotic behavior occurs for the $B_0$ solutions while a different class of asymptotic behavior occurs for the solutions with $B\ne 0$. Consequently, the asymptotic behavior for the $B_0$ solutions is not stable in the full set of $T^2$-Symmetric spacetimes. The graphs of the numerical solutions presented below illustrate the different behavior of the $B_0$ asymptotics compared to the asymptotic behavior for solutions with $B\ne 0$.

Although the lack of control of the quantity $\int_{S^1} e^\rho V_\tau d\theta$ for the $B\ne 0$ solutions prevents the use of the analysis of \cite{MR3513138} and \cite{MR4062458} for mathematically proving that the asymptotic behavior seen in our numerical simulations occurs necessarily, we do provide in Section \ref{future}  a rough framework for how one might mathematically verify the asymptotic behavior seen in the $B\ne 0$ solutions. In future work, we hope to be able to determine whether this approach or an alternative one is likely to be successful.

Since the asymptotics of the $B = 0 $ and $B \ne 0$ spacetimes are different, most of the figures displaying numerical results show the results of simulations for both cases. 
In the absence of mathematical proof, the numerical simulations presented in this paper provide the only currently available way to infer the expected asymptotic behavior in the $B \ne 0$ case.

\section{Numerical simulations}

\subsection{Numerical solution of the constraints in these coordinates.}\label{numerical}

The coordinate form for the metric \eqref{coords} is used in \cite{MR4062458} and is particularly useful in describing the qualitative behavior in the expanding direction.
The numerical work we discuss here, however, makes use of the quantities \(P,\pi_P,Q,\pi_Q,\lambda,\pi_\lambda\) of \cite{MR1858721}.
The choice of these variables is only a matter of convenience; our code for this project is based on code previously developed by the first author using the quantities in \cite{MR1858721}.
For the discussion here it is only important to note that in these coordinates, the momentum constraint \eqref{constraint} takes the form
\begin{align}
    P_\theta \pi_P + Q_\theta \pi_Q + \lambda_\theta \pi_\lambda = 0 \label{pi constraint}
\end{align}
subject to the condition \(\pi_\lambda>0\).\footnote{
Note that comparing the coordinates of the present paper to those used in \cite{MR1858721}, we find that \(e^\rho = 2 \pi_\lambda\).
The coordinates of \cite{MR1858721} are derived from those in \cite{MR1474313} where the metric has the form
\begin{align}
    g= e^{2(\nu - U)}\left( - \alpha \, dt^2 + d\theta^2 \right) + \cdots
\end{align}
and \(\pi_\lambda := \frac{1}{2}\alpha^{-1/2}\).
}
We refer to \cite{MR1858721} for a more complete discussion of these variables, including the evolution equations. (However, we do note that the time coordinate in this paper is the negative of the time coordinate in \cite{MR1858721}.)
All of our numerical simulations 
 use \(K = 1\).
In view of Eq. \eqref{lsub} this assumption is generic; for every non-Gowdy, \(T^2\)-Symmetric solution with a value of \(K \ne 1\) there is a solution with \(K=1\) which may be obtained by adding a constant to \(\widehat l\).

In order to numerically simulate solutions of the
Einstein field equations, we must begin with initial data that solve the constraint equations.
Furthermore, if one wants to find reliable information about generic solutions of the field equations, one must begin with solutions of the constraints which are as generic as possible.
The issue of generic solutions of the numerically simulated constraints is fundamental for this project.
Previous mathematical results, as discussed in Section \ref{math results}, have considered only the subfamilies of the Gowdy, the polarized, the pseudo-homogeneous solutions, and the $B=0$ solutions, all of which are nongeneric in the class of all \(T^2\)-Symmetric Einstein flows on \(T^3\).
Furthermore, previous simulations that have aimed to simulate the full Einstein flow in these coordinates without any simplifying assumptions have inadvertently imposed the condition \(B=0\). From the form of the momentum constraint \eqref{pi constraint}, it is clear that an obvious way to solve it is to assume that one term in each of the three pairs of terms vanishes initially. The definition of $B$ makes it clear that this is equivalent to setting \(B=0\). It had not been realized prior to the discussion in \cite{MR4062458} that these simplifying assumptions, which imply $B=0$, restrict these simulations to the study of non-generic initial data, and consequently to simulations of non-generic subfamilies of the full set of $T^2$-Symmetric solutions.
\mnote{JIM?: DONE by BKB. Explain why \(A\) isn't meaningful}
\mnote{BEVERLY?: Note the actual ranges for \(A,B\) in the data}
This example shows us that many relatively simple approaches to solving the constraints are likely to generate non-generic initial data sets.

To avoid this restriction, we have generated initial data solutions of the constraints using a spectral method,
\mnote{BEVERLY: Reference for spectral method}
expanding all six of the functions appearing in Eq. \eqref{pi constraint} to some finite Fourier order.
We are then able to choose as free data some number of the Fourier coefficients so that Eq. \eqref{pi constraint} is a nondegenerate linear equation in the remaining coefficients for which we solve.
Finally, if needed, we shift \(\pi_\lambda\) vertically by a chosen amount so that it is strictly positive.

Some of the initial data sets we have obtained have resulted in evolutions which have poor numerical convergence.
As a consequence, we have manually eliminated these initial data sets and their simulated evolutions in making our predictions regarding the behavior of generic \(T^2\)-Symmetric solutions.
Subsequently, we have discerned that the solutions with poor numerical convergence can be distinguished from the solutions with good numerical convergence using the norm:
\begin{align}
    M(P,\pi_P,Q,\pi_Q,\lambda,\pi_\lambda) = \max_{\theta \in S^1} \max_{f \in \{P_\theta , \pi_P , Q_\theta , \pi_Q , \lambda_\theta , \pi_\lambda\}} f^2.
\end{align}
Initial data sets with large \(M\)-norm involve functions which (because of the way we generate random initial data) are likely to have one of the ``momentum terms'' much larger than the rest.
When evolving numerically, the errors introduced are of the magnitude of the largest momentum term, and are larger than would otherwise be the case in comparison to the other momentum terms.\footnote{The issue is the the same as arises in numerically solving, say, the equation \(\ddot f = - f\) with initial data \(f(0) = \gamma\) and \(\dot f(0) = \gamma^{10}\), where \(\gamma \gg 1\).
The errors are proportional to \(\gamma^{10}\).}
Since the $M$-norm effectively distinguishes ``good initial data sets'' from ``bad initial data sets'', we expect that, in any future numerical simulations, we would restrict our studies to initial data sets with all of the momenta terms of approximately the same order. We note that only a small fraction of the numerical evolutions of our initial data sets exhibit poor numerical behavior. We do not know if this reflects any underlying mathematical properties of the set of generic $T^2$-Symmetric spacetimes.

Fig. \ref{ABdata} shows the initial data for the models explored in this study in terms of the integrated constants $A$ (from \eqref{a def}) and $B$ (from \eqref{b def}). The value of $A$ is more or less unrestricted while $B$ lies in the range $[-1,1]$. Note that these constants are obtained after the procedure described here is followed.  Because $B \ne 0$ represents an instability in the Ringstr{\"o}m sense \cite{MR3186493}, we believe that a restriction on the range of $B$ does not imply a loss of generality. 

\begin{figure}
    \centering
    \includegraphics[scale=.5]{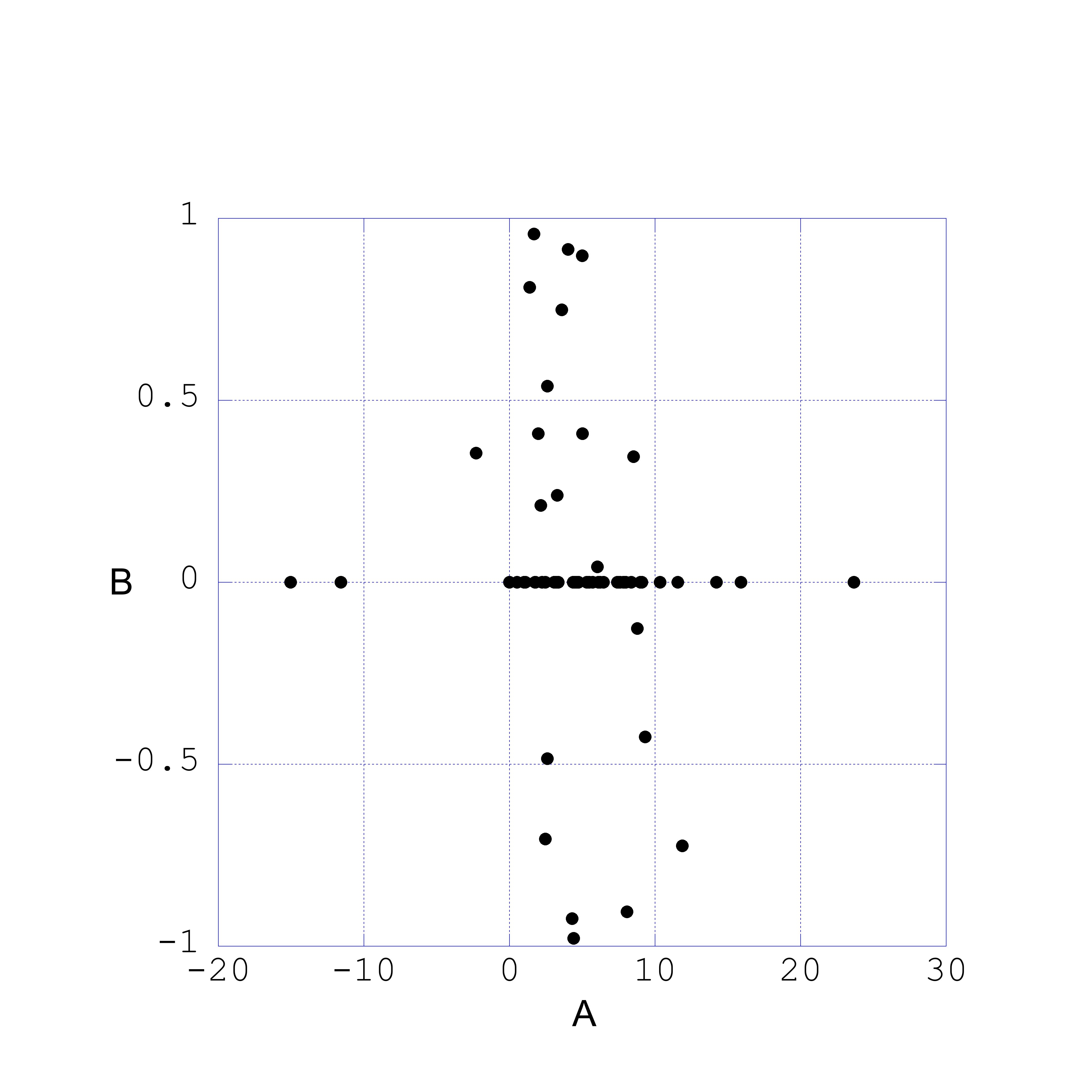}
    \caption{The values of the constants $A$ and $B$ for the initial data sets we have numerically constructed and then evolved are indicated here.}
    \label{ABdata}
\end{figure}




\subsection{	Numerical Methods.}

Numerical evolution, using Eqs. \eqref{v transport}--\eqref{rho-evol} of $T^2$-Symmetric spacetimes from initial data sets constructed as described above, has been carried out using original code implementing finite differences in the space dimension and iterated Crank--Nicholson (ICN) in the time direction \cite{2000PhRvD..61h7501T}.
The code is based on code previously used to simulate the contracting direction of \(T^2\)-Symmetric solutions \cite{MR1858721}.
This code has been reimplemented in a modern, functional language OCaml \cite{ocaml}, and certain improvements have been incorporated.\footnote{The code is available at the third author's github \url{https://github.com/fNBU/BIL-numerical-T2} where all the initial data and data for the time evolutions can also be found.}
Finite differences and ICN are relatively naive methods. However, the topology of the space domain and the fact that the equations are quasilinear prevent standard packages from being used.
Evolution has been done using variables from \cite{MR1858721}.
We believe that the simulations we have carried out are accurate because
\begin {enumerate*} [label=\itshape\roman*\upshape)]
\item the numerical methods used are known to converge in a wide variety of examples,
\item the momentum constraint \eqref{pi constraint}, which in exact solutions remains zero throughout the evolution, stays reasonably close to zero almost everywhere during the numerical simulation (see Figs. \ref{cavg}, \ref{cmax} , and \ref{3Dconstraint});
\item the integrated quantities $B$ and $A$, which in exact solutions remain constant throughout the evolution, stay reasonably close to constant during the evolution with the deviation decreasing as the spatial resolution increases (see Figs. \ref{Bfig} and \ref{Afig});
\item the solutions converge at the expected rate upon refinement of the spatial mesh, and
\item for the polarized and the $B=0$ subfamilies, the behavior of the numerical simulations agrees with that which has been proven mathematically.
\end {enumerate*} 
 Figures \ref{cavg} and \ref{cmax} illustrate the behavior of the momentum constraint $\cal {C} $ which is equal to the left-hand side of  \eqref{pi constraint} in examples of our numerical simulations.\footnote{The designation $\cal {C}$ has been chosen to avoid confusion with the Gowdy constant $C$ mentioned above.}
 
 Figure \ref{cavg} plots the evolution (up to the time $\tau = 20$, which essentially corresponds to the asymptotic expansion time) of the absolute value of the spatial average of $\cal {C}$ for spatial resolutions of $\theta=2\pi/n$ for $n= 512, \ 1024, \ 2048$, and $4096$. While the absolute value of this spatial average of $\cal {C}$ does grow in this example, we note that its value decreases for the improved resolutions, which gives us confidence that our simulations are reasonably accurate. Figure \ref{cmax} plots the evolution of the maximum value of the absolute value of $\cal {C}$, again showing that, for increased spatial resolutions, the values decrease, indicating accuracy for our numerical simulations. We note that in this work, the Hamiltonian constraint is treated as an evolution equation \eqref{l-evol} and thus is automatically satisfied.

Figures \ref{Bfig} and \ref{Afig} compare the behavior of the integrated quantities $B$ and $A$ in simulations of differing spatial resolution. Specifically in these figures, we illustrate the evolution of $\Delta B$ and $\Delta A$ – the differences of these quantities from their initial values for representative simulations. As noted above, for exact solutions, $B$ and $A$ are constants. As with $\cal {C}$ in Figs. \ref{cavg} and \ref{cmax}, in these figures we plot the evolution of $\Delta B$ and $\Delta A$ for a succession of increasingly fine spatial resolutions. We see in these figures that the evolution of $\Delta B$ becomes
closer and closer to zero as the spatial mesh becomes finer and finer. This does not occur for $\Delta A$; indeed, we see that after $\tau= 15$, this quantity behaves worse for the highest spatial resolution than it does for the next highest spatial resolution. (This misbehavior occurs in both the $B=0$ and the $B\ne 0$ cases.) 
Careful examination of the behavior of the individual terms in $A$, shows that they separately become several orders of magnitude larger than the actual value of $A$ so that the required subtraction becomes delicate. Furthermore, as $\tau$ increases, the Courant condition for these wave equations causes the timestep to become very small, magnifying the problem. This behavior is characteristic of stiff PDEs and may possibly be cured by rewriting the differenced equations. However, we note that the actual value of $A$ is a derived quantity that is not required for any computation. We therefore assume that constraint convergence validates the accuracy of the numerics and of the asymptotics we obtain therefrom for the generic case. We emphasize that asymptotics which we obtain here numerically are consistent with the mathematically proven asymptotics for the $B=0$ subfamily.

 \begin{figure}
     \centering
     \includegraphics[scale=.35]{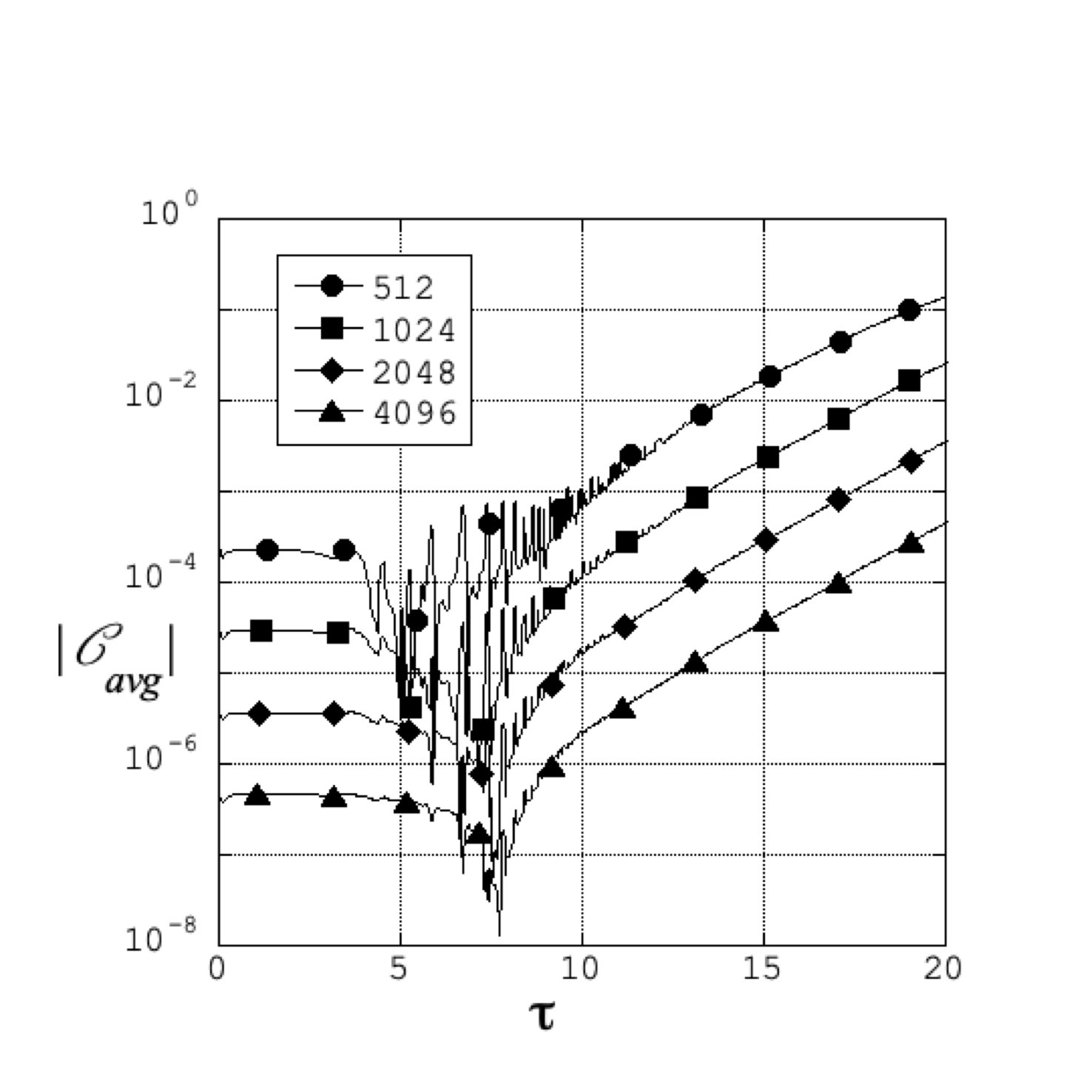}
     \caption{Convergence of the spatial average of the momentum constraint, $\cal {C}$, with finer spatial resolution. 
     The absolute value is used to allow a logarithmic vertical scale. The results are shown for a representative numerical simulation over the time interval $\tau \in [0,20]$, with the spatial resolution $\theta = 2 \pi/n$ for $n = 512, 1024, 2048, 4096$. We see that as the spatial resolution becomes finer, the evolution of this quantity becomes closer to zero, the value expected for an exact solution. }
     \label{cavg}
 \end{figure}
  \begin{figure}
     \centering
     \includegraphics[scale=.35]{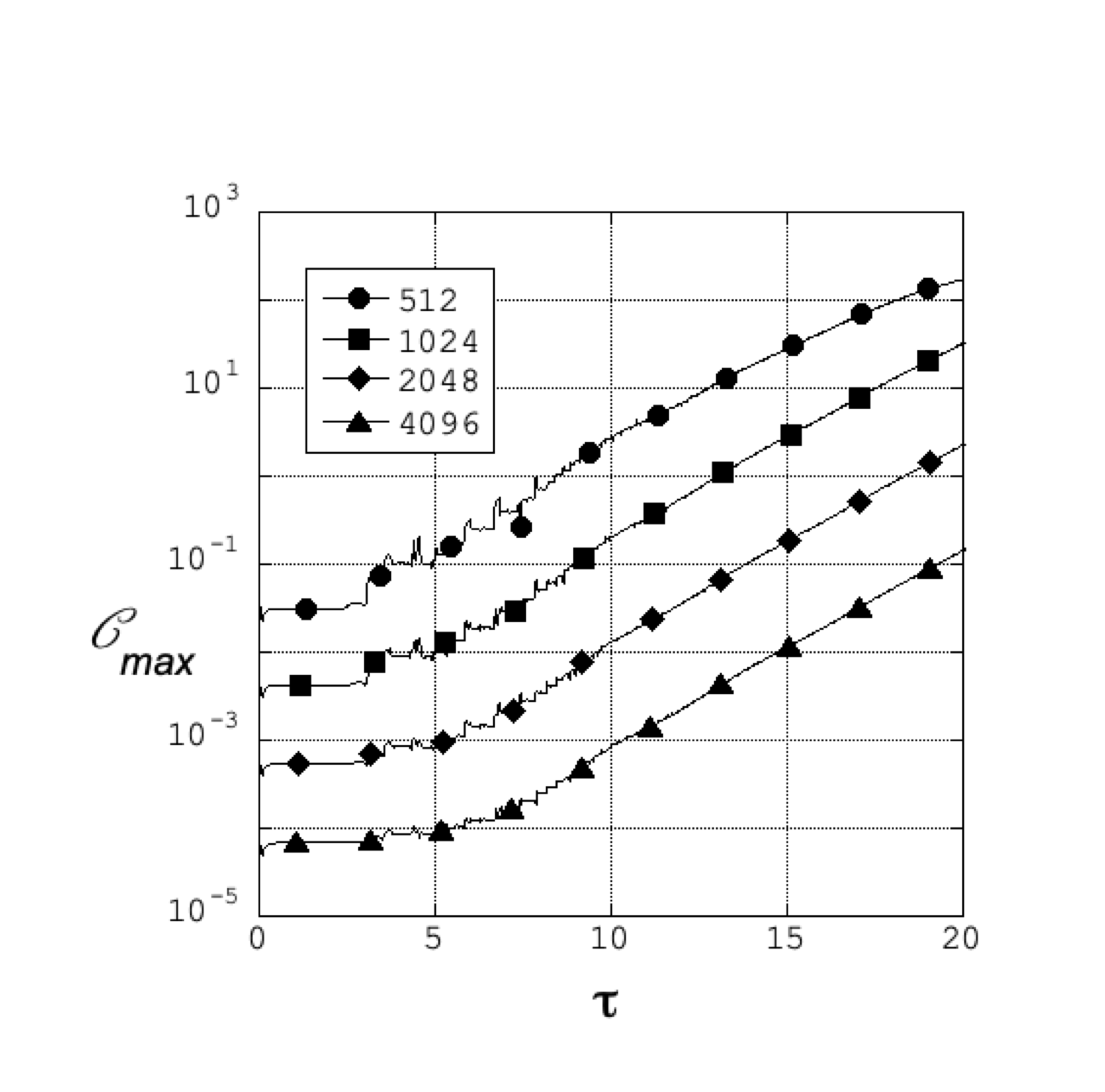}
     \caption{The same as Fig. \ref{cavg} but for the maximum of the absolute value of the momentum constraint.}
     \label{cmax}
 \end{figure}

  \begin{figure}
     \centering
     \includegraphics[scale=.35]{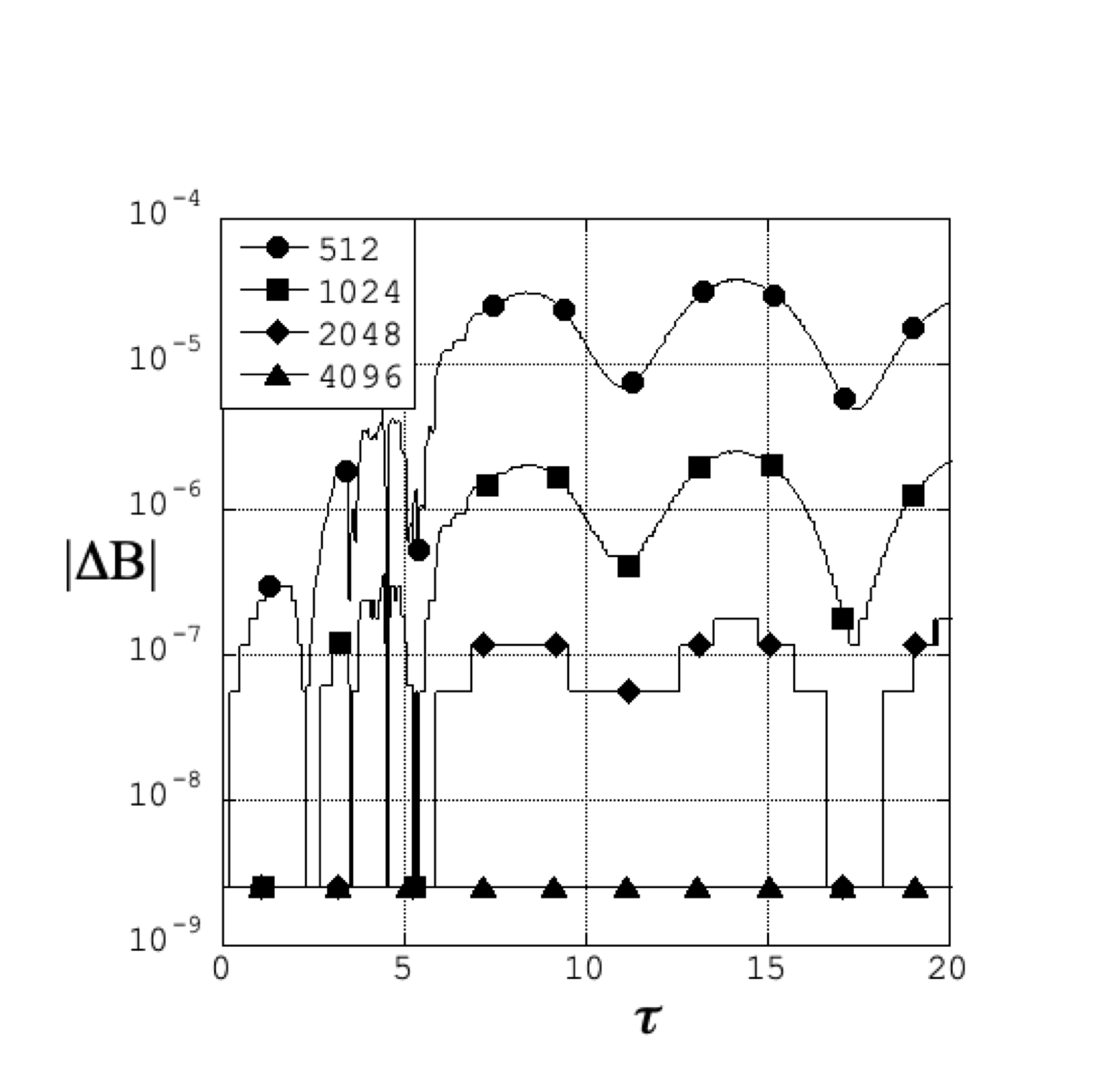}
     \caption{We use $\Delta B (\tau)$ to designate the difference of $B$ from its initial value for a representative numerical simulation over the time interval $\tau \in [0,20]$, with the spatial resolution $\theta=2\pi/n$ for $n$ the number of spatial points. As $n$ increases, $\Delta B$  converges to  zero, which would be the case for an exact solution.
}
     \label{Bfig}
 \end{figure}
 
  \begin{figure}
     \centering
     \includegraphics[scale=.35]{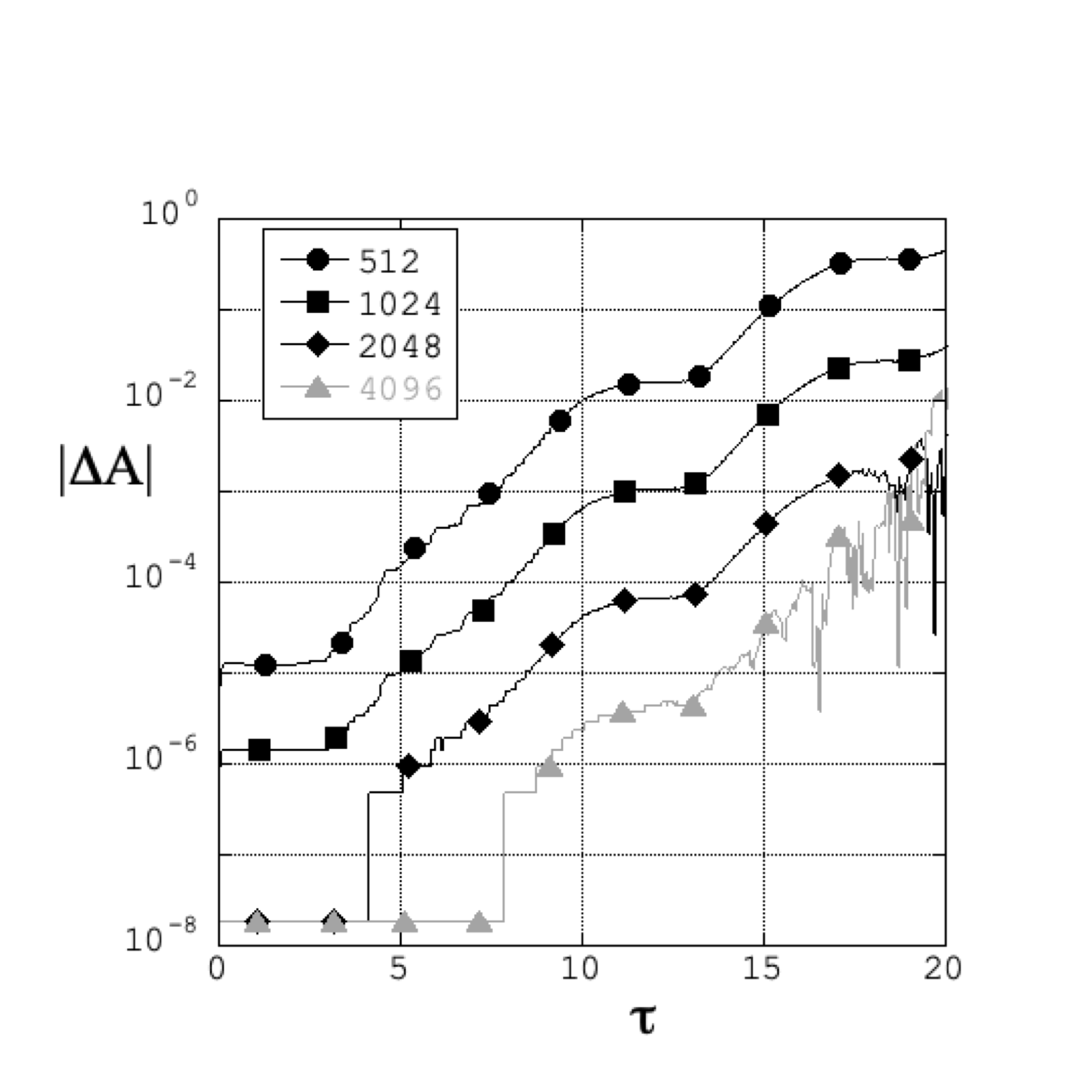}
     \caption{We use $\Delta A$ to designate the difference of $A$ from its initial value for a representative numerical simulation over the time interval $\tau \in [0,20]$ with the spatial resolution $\theta=2\pi/n$. For $\tau < 12$, convergence as $n$ increases is comparable to that shown for the momentum constraint. However, the convergence breaks down toward the end of the simulations and is worse for the largest values of $n$. (For this reason, the curve for that resolution is shown in gray.) A possible explanation is discussed in the text. }
     \label{Afig}
 \end{figure}
 
 \begin{figure}
    \centering
    \includegraphics[width=\textwidth]{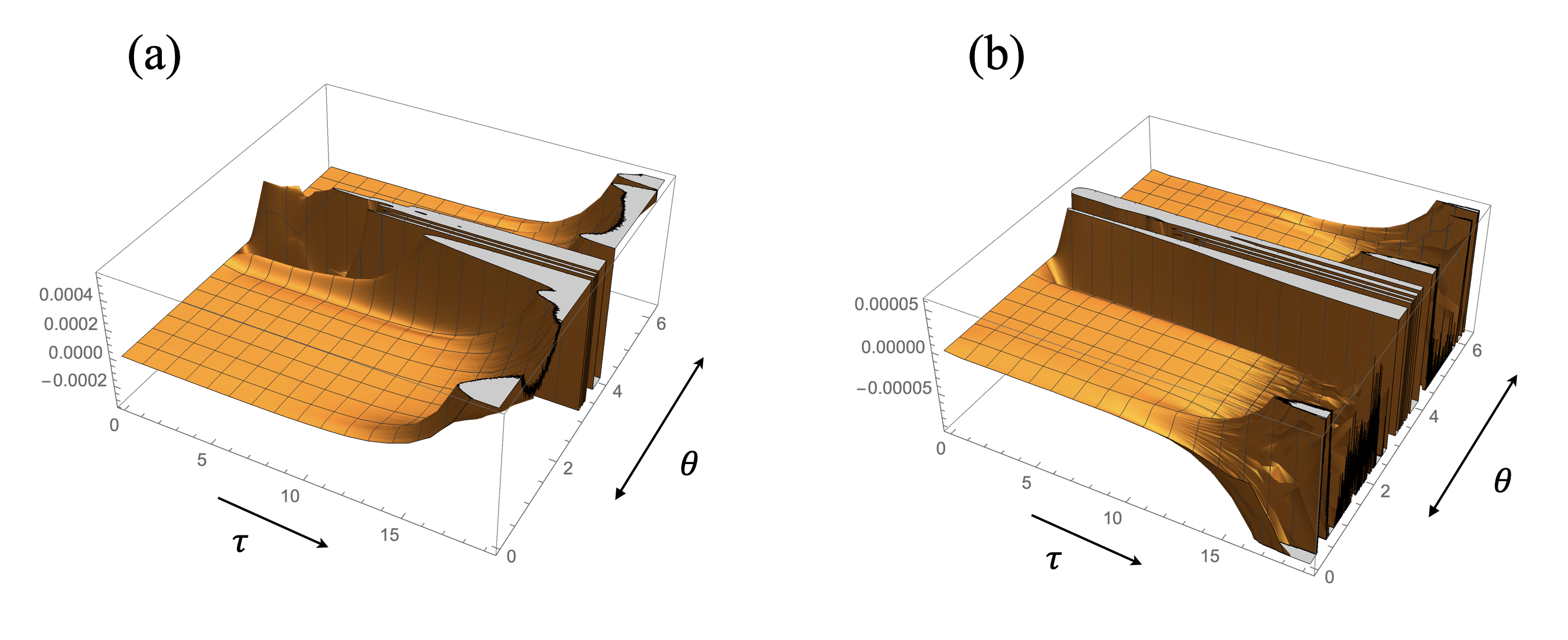}
    \caption{An example \cite{Mathematica} of the evaluation of the momentum constraint \eqref{c def} over the entire computational spacetime $\{\theta,\tau\}$. (a) The $B=0$ case. (b) The $B \ne 0$ case. Both cases have a resolution of 2048 spatial points.}
    \label{3Dconstraint}
\end{figure}

\section{Results of Simulations and Implied Asymptotic Behavior \label{impliedasymptotic}}

\subsection{Dominant Asymptotic Behavior of Metric Coefficients and Their Conjugate Momenta }


In this subsection, we discuss the asymptotic behavior as $\tau \to \infty$ that our numerical simulations imply for the metric coefficients and their conjugate momenta for a variety of solutions, both of subfamilies as well as for the generic case. 
In carrying out the numerical simulations of the evolutions of these spacetimes, we focus our attention on increasing values of $\tau$, which is the expanding direction. From these simulations, we infer the asymptotic behavior for very large values of $\tau$.

Table \ref{tab} lists the asymptotic behavior which our simulations indicate for three classes of $T^2$-Symmetric solutions: the generic solutions, the $B=0$ solutions, and the polarized solutions. These behaviors, expressed in terms of expansions in $\tau$,  are consistent with the numerics, but have not been proved except in special cases.
Some of proofs of the estimates in the \(Q \equiv 0\) (polarized) column of Table \ref{tab} require a weaker form of the estimates -- e.g., imposing a smallness condition -- to carry out the proofs. (See \cite{MR3513138}.)
Some of the estimates in the \(B=0\) column have been proven in \cite{MR4062458} although again, some of the proofs rely on a weaker form of the estimates.
To our knowledge, none of the estimates in the \(B\ne 0\) column have been proved.
\begin{table}
    \centering
    \setlength{\tabcolsep}{12pt}
    \renewcommand{\arraystretch}{1.5}
    \begin{tabular}{r||l|l|l|r}
        &	\multicolumn{1}{|c|}{\(B \ne 0\)}													&	\multicolumn{1}{||c|}{\(B = 0\)}	&	\multicolumn{1}{||c|}{\(Q \equiv 0\)}	&					\\\hline\hline
        \(\rho\)								=&	 \multicolumn{3}{c|}{\(\frac{1}{2} \tau + \rho_\infty +  O(e^{-\tau/4})\)}		& \dag \ddag 	\\ \hline
        \(\rho_\theta\)					    	=&	 \multicolumn{3}{c|}{\(O(1)\)} 														&	\dag \ddag	\\\hline
        \(\rho_\tau\)							=&	 \multicolumn{3}{c|}{\(\frac{1}{2} + O(e^{-\tau/4})\)}							&	\dag \ddag	\\\hline
        \(l		\)							=&	\( \ln \frac{1}{2} + O(e^{-\tau/4})			\)									&	 \multicolumn{2}{||c|}{\(\ln \frac{1}{2} + O(e^{-\tau/8})\)}			&					\\\hline
        \(l_\theta\)							=&	 \multicolumn{3}{c|}{\(O(e^{-\tau/2})\)}											&	\dag \ddag	\\\hline
        \(l_\tau\)								=&	 \multicolumn{3}{c|}{\(O(1)\)}														&	\dag \ddag	\\\hline
        \(V\)									=&	 \(\frac{1}{2} \tau+ V_\infty  +  O(e^{-\tau/4})\)									&	 \multicolumn{2}{||c|}{\(V_\infty  + O(e^{-\tau/2})\)}					&	\dag \ddag		\\\hline
        \(V_\theta\)							=&	 \multicolumn{3}{c|}{\(O(e^{-\tau/2})\)}											&	\dag \ddag	\\\hline
        \(V_\tau\)								=&	 \multicolumn{3}{c|}{\(O(1)\)}														&	\dag \ddag	\\\hline
        \(e^{-\tau/2} Q\)						=&	 \({\cal{Q}}_\infty + O(e^{-\tau/4} )\)														&	 \multicolumn{1}{||l||}{\(O(1)\)}											&	 \(0\)				\\\hline
        \(\partial_\theta(e^{-\tau/2}Q)\)	=&	 \(O(e^{-\tau/2}	)\)																	&	 \multicolumn{1}{||l||}{\(O(1)\)	}											&	 \(0\)				\\\hline
        \(\partial_\tau(e^{-\tau/2}Q)\)		=&	 \(O(1)\)																				&	 \multicolumn{1}{||l||}{\(O(e^{\tau/2})\)}									&	 0	
    \end{tabular}
    \caption{Asymptotic values of the metric components as \(\tau \to \infty\) for polarized, \(B=0\), and generic ($B \ne 0$) solutions.
    All of the expansions for the class \(B \ne 0\) are conjectural results based on our simulations.
    The results labeled with a {\ddag} have been proven in \cite{MR4062458} for sufficiently small \(B=0\) solutions.
    The results labeled with a {\dag} have been proven in \cite{MR3513138} for sufficiently small polarized solutions.
    These asymptotic values are observed in our simulations without regard to a smallness condition or other, similar restriction, and are conjectured to hold generically for that class. Separate asymptotics of $l$ and $Q$ have not been proven. }\label{tab}
\end{table}

We note that smallness criteria on the estimates are imposed in Refs. \cite{MR3513138} and \cite{MR4062458} to allow the respective analyses to proceed by allowing small terms to be dropped as needed. In effect, the limit is placed on initial data to ensure that deviations from the proven asymptotic behavior are never too large. See \cite{MR3513138} and \cite{MR4062458} for the specific criteria used.

The smallness criteria needed for proofs in \cite{MR3513138} and \cite{MR4062458} are quite different from the limitations of our numerical methods which fail for initial data with momenta that are not all approximately the same magnitude. Experience has shown that failed simulations violate this restriction. We believe this issue is solely numerical, resulting if the differenced evolution equations become stiff. Fixing the problem will be left to future work. Although we cannot prove it, we believe that the failed simulations do not indicate yet another class of previously missed solutions.

Note that the numerically inferred expansions of $V$ and $e^{-\tau/2} Q$ in Table \ref{tab} and the bounds on their space derivatives indicate that \( V_\infty, Q_\infty\) are actual spatial constants.
This is not the case for \(\rho_\infty\); \(\rho\) generically converges to a non-constant function of $\theta$.

Let us demonstrate how we conclude that these estimates are optimal.\footnote{The presence of the factor $e^{-\tau/2}$ in the expressions given for the asymptotics of $Q$ is needed because $Q$ does not asymptotically decay (in contrast, e.g., to $V$) but rather grows as $e^{\tau/2}$. This behavior can be seen in Fig. \ref{proof}.}
We select from the many solutions we have simulated six (three each of the \(B=0\) and \(B\ne 0\) cases) which have the best numerical convergence.
For each of these six simulations, we generate a plot of, for example, the scalar function of time
\begin{align}
    X_R :=  \max_{S^1} e^{R \tau}\left| \rho - \rho_{\widetilde \infty} - \frac{\tau}{2} \right| \label{bounding example}
\end{align}
for various values of the constant \(R\) until we find a value of $R$ which results in $X_R$ being \(O(1)\) but non-decaying as \(\tau \to \infty\).
\begin{rem}\label{inftytilde}
    In the simulations, we cannot actually set \(\left. \rho \right|_{\tau=\infty}\) equal to \(\rho_\infty\), the mathematical limit, since our simulations do not extend to \(\tau = \infty\).
    Instead, we use the value of \(\rho\) at a fixed time \( \tau = \widetilde \infty=20\) near the end of the simulation.
    This is well within the asymptotic regime.
    We use the notation \(f_{\widetilde \infty} := \left. f \right|_{\tau = \widetilde \infty} \).
\end{rem}

\begin{figure}
    \centering
    \includegraphics[scale=.3]{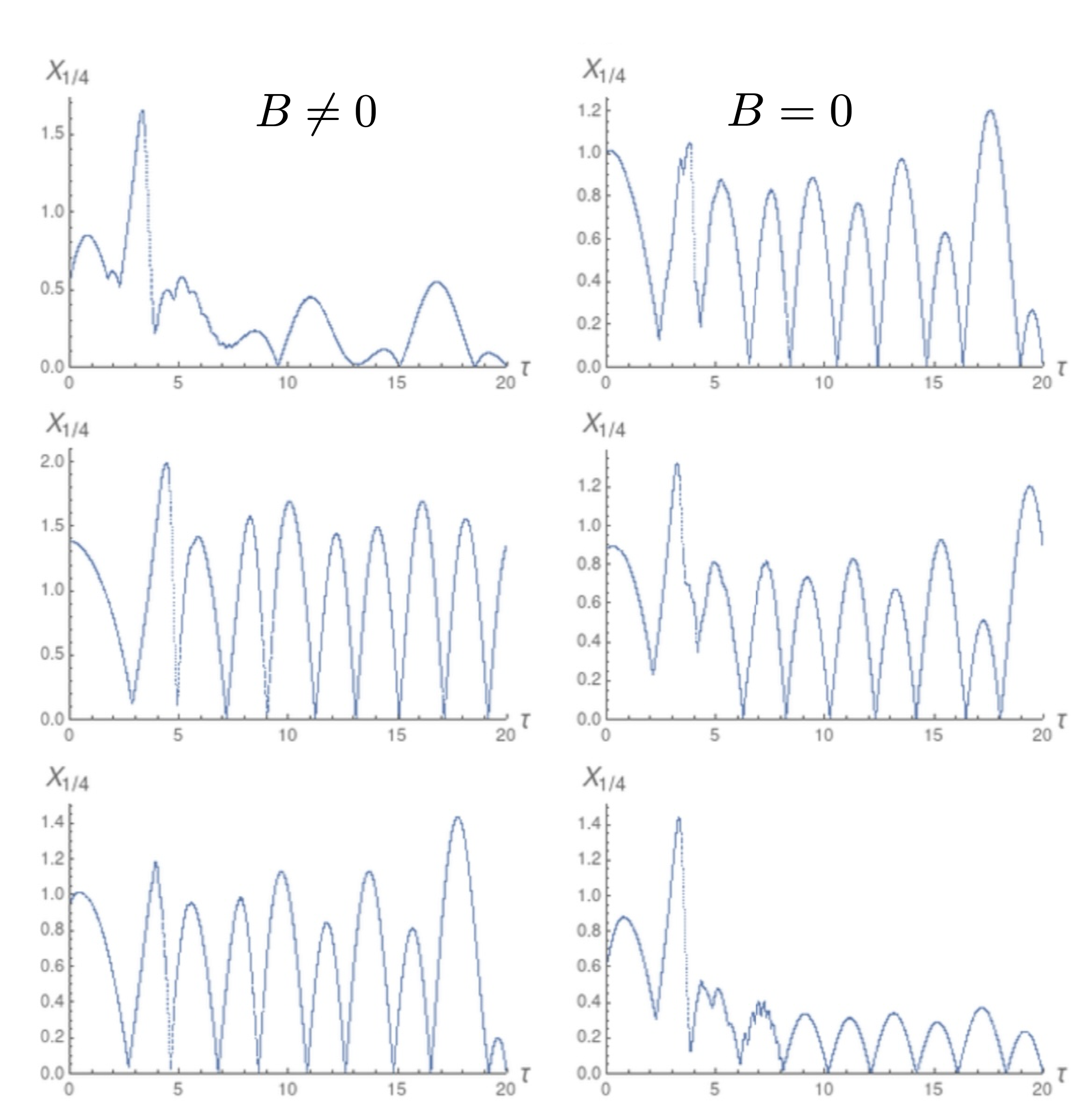}
    \caption{Plots of the time evolution of the function $X_{1/4} (\tau)$ for numerical simulations of three different solutions for the generic $T^2$-Symmetric class as well as three different solutions in the $B=0$ subfamily.  $X_R$ is defined in \eqref{bounding example}. The different subclasses are shown horizontally and the different simulations vertically. }\label{rhoplts}
\end{figure}

In Fig. \ref{rhoplts} we present plots that generate the
bounding value for the quantity \(X_R\) as defined in equation \eqref{bounding example}. The objective of these plots is to indicate that all models for representative simulations approach a constant envelope as $\tau \to 20$ indicating (i) that the asymptotic regime is likely to have been reached and (ii) that $X_R = X_{1/4}$ demonstrates the correct power law. One might think that a different value of \(R\) would give a better bound.
In Fig. \ref{rhopltscomparison}, overlays of $X_R$ for $R = 1/8,\ 1/4,\ 1/2$ are shown for both the $B=0$ and $B \ne 0$ classes of models to demonstrate that only the $R = 1/4$ value is likely to lead to convergence as $\tau \to \infty$. This application of bounding power laws allows inference of the asymptotic behavior from numerical simulations.

\begin{figure}
    \centering
    \includegraphics[width=\textwidth]{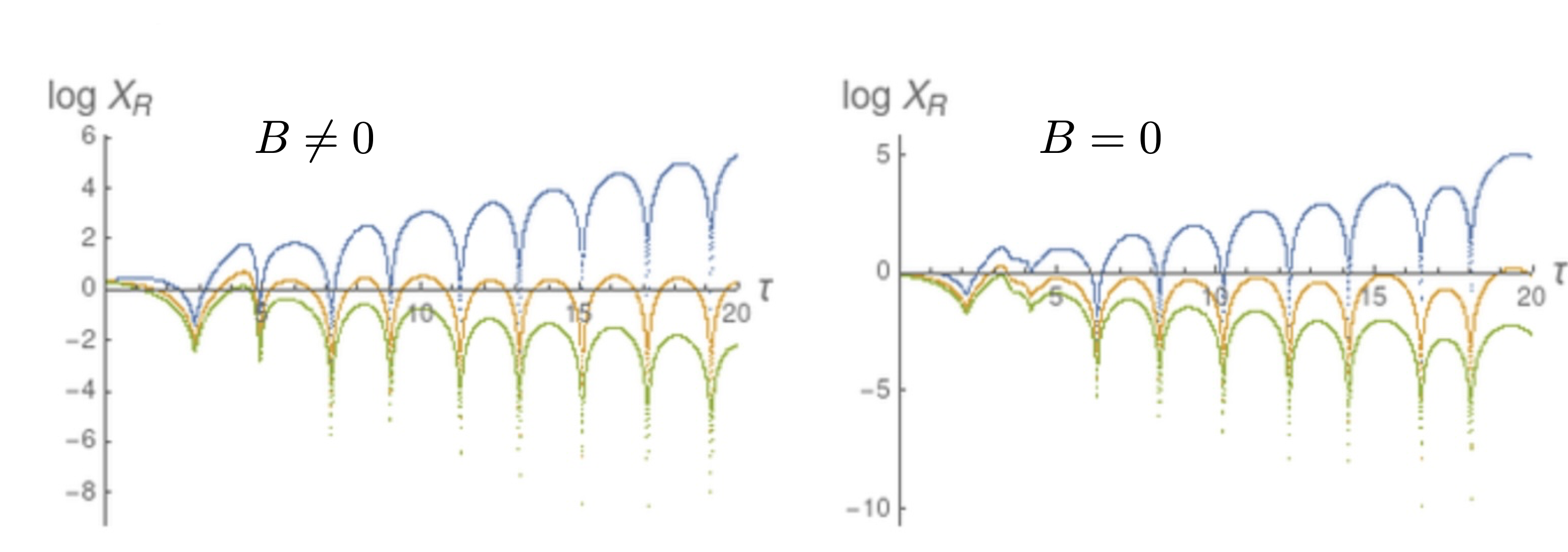}
    \caption{Plots of the evolution of the functions \(\ln X_R\) for different values of \(R\) for the two classes of solutions. \(R=1/2\) is in blue, \(R = 1/4\) is in orange, and \(R = 1/8\) is in green. These plots demonstrate that, with $R= 1/4$ for both classes, $X_R$ neither grows nor decays as it does for $R= 1/2$ or $R= 1/8$ respectively.}\label{rhopltscomparison}
\end{figure}

To see that \(V\) has a different rate of expansion in the \(B=0\) and \(B\ne 0\) cases, assume that \(V = V_\infty + O(e^{-\tau/2})\) and plot 
\begin{align}
    Z_R := \max_{S^1} e^{ R\tau}\left| V - V_{\widetilde \infty}  \right|
    \label{Ydef}
\end{align}
with \(R = 1/2\).
This plot is presented in Figure \ref{vplt1}.
\mnote{?: High frequency damping in \(V\) plots}
\mnote{?: say something about how the high frequency oscillations differ in the \(B=0\) and \(B\ne 0\) cases}
\mnote{JIM: Also, I think it would be helpful to comment further on the distinctions between these graphs for the \(B\neq 0\), the \(B=0\), and the polarized cases.}
\mnote{?: say something about the large scale bounce in \(V\) in the \(B \ne 0 \) case.}
\begin{figure}
    \centering
    \includegraphics[scale=.25]{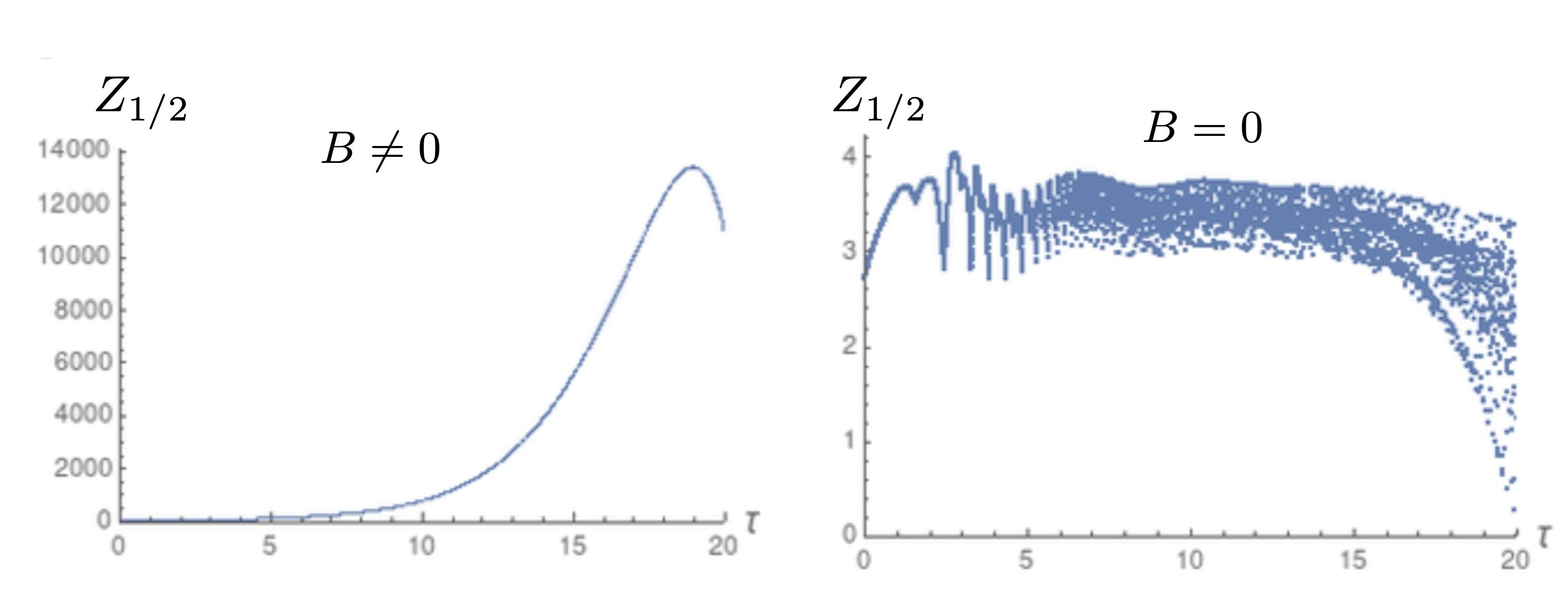}
    \caption{Plots of the evolution of the function \(Z_{1/2}\) for numerical simulations of a solution for the two classes of solutions. Notice that \(Z_{1/2}\) is \(O(1)\) only if \(B = 0\).}\label{vplt1}
\end{figure}

\begin{rem}
It may seem odd that all the plots in Figure \ref{vplt1} tend to zero at the end of the simulation.
This is an artifact of the fact that our simulation does not extend to \(\tau =\infty\).
Thus, by construction, \(\left| V - V_{\widetilde \infty}  \right| \to 0\) at the end of the simulation.
See Remark \ref{inftytilde}.
\end{rem}

\begin{rem}
    It may also seem odd that the plots in Figure \ref{vplt1} are ``fuzzy''.
    This feature is related to the fact that the natural volume form in the \(\theta\)-direction is \(e^{\rho - \tau}\, d\theta\).
    Mean values of the form \(\int_{S^1} f \, d\theta\) and sup norms can have highly oscillating behavior because these are not physical quantities.
    Rather, one should compute \(\int_{S^1} f \, e^{\rho - \tau}\, d\theta\) to get a smooth, physical quantity.
    Nonetheless sup norms are useful for establishing low-order behavior as \(\tau \to \infty\). It should also be kept in mind that the ``fuzziness'' occurs on a scale of order unity or smaller. If the vertical scale we use in the figure were one or two orders of magnitude larger, the fuzziness, while still present, would be invisible in the figure.
 \end{rem}

Notice that in Fig. \ref{vplt1}, the plot for the $B=0$ solution is {\cal{O}}(1), but the plot for the $B\ne 0$ solution diverges exponentially.
If we instead assume that the bound is of the form \(V = V_\infty + \frac{\tau}{2} + O(e^{-R \tau})\) and plot \(\widehat Z_R := \max_{S^1} e^{ R\tau}\left| V - V_{\widetilde \infty} - \frac{\tau}{2}  \right|\) with \(R = 1/4\), we obtain the plots in Figure \ref{vplt2}.
\begin{figure}
    \centering
    \includegraphics[scale=.25]{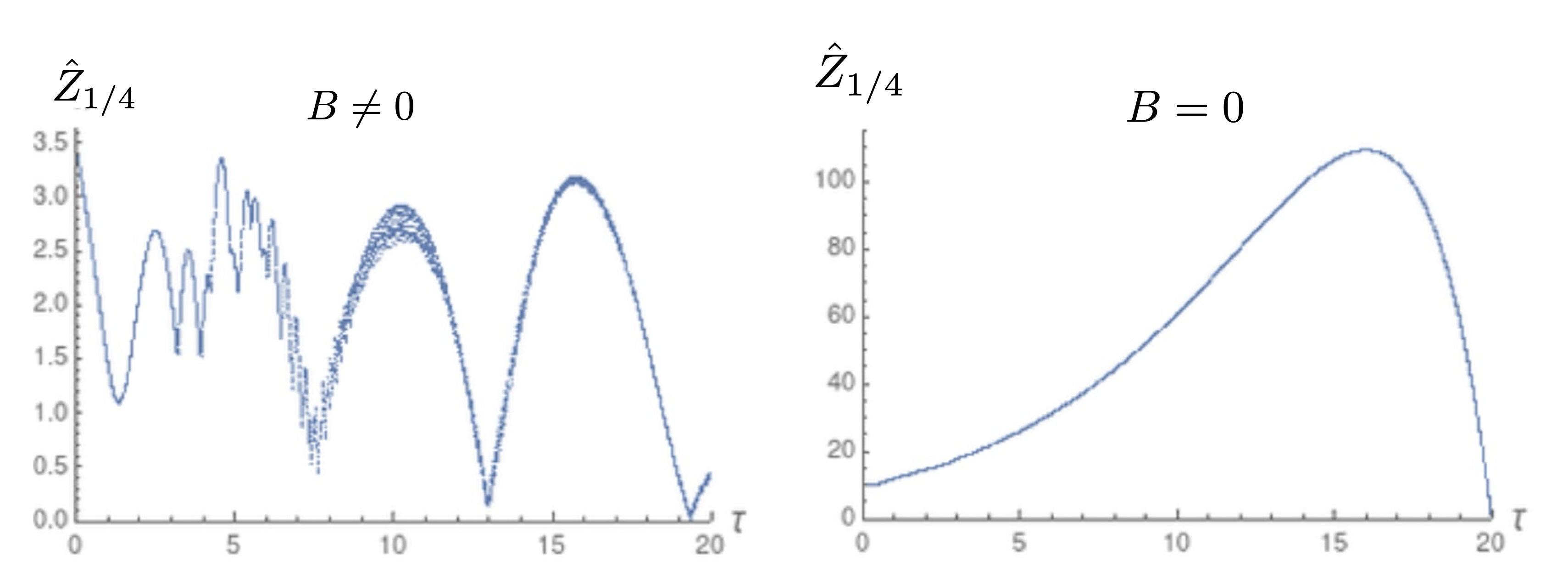}
    \caption{Plots of the evolution of the function \(\widehat Z_{1/4}\) for numerical simulations of a solution of the two classes of solutions. Notice that \(\widehat Z_{1/4}\) is \(O(1)\) only in the \(B\ne 0\) case. This means that this is the only case where the presumed power law $R = 1/4$ is correct for the generic case but incorrect for the $B=0$ case. As mentioned above, the differing vertical scales allow oscillations to be visible if the scale is order unity but invisible for scales extending to $100$ or greater. }\label{vplt2}
\end{figure}
Notice that the the generic example in Fig. \ref{vplt1} and the $B=0$ case in Fig. \ref{vplt2} appear to grow exponentially and then fall to zero. As discussed following Eq. \eqref{Ydef}, this is an artifact due to the simulation not extending in time to $\infty$.




\subsection{An ODE Approach for Determining the Asymptotic Behavior of the Means of the Metric Coefficients \label{future}}
Given that \(V\) and \(e^{-\tau/2}Q\) converge to spatial constants, it is natural to form a system of ordinary differential equations (modulo error terms) which govern the evolution of the spatial means of these quantities. 
This is the technique used for proving certain results in \cite{MR3513138} for polarised \(T^2\)-Symmetric solutions and in our paper \cite{MR4062458} for the \(B = 0\) \(T^2\)-Symmetric solutions which are not necessarily polarised.

We can see from the numerics that in the general (\(B\ne 0\)) case it is \(|V - V_\infty - \frac{\tau}{2}|\) that decays, not \(|V - V_\infty|\).
It then seems natural to study the quantity \(|V - \frac{\tau}{2}|\) instead.
We have not yet been able to determine the future behavior of \(V\) in the \(B\ne0\) case, partially due to the weaker decay of this quantity.
Let us describe how an argument determining this future behavior might be expected to proceed.

We expect such a proof, if it were to follow the basic idea of the approach used in \cite{MR4062458}, to proceed in four steps as follows:
\begin{enumerate}
    \item 
    Denoting the right side of Eq. \eqref{l-evol} by \(J\), one would like to form an ordinary differential equation system involving the spatial averages:
    \begin{align} \label{pi def}
        \Pi:= \int_{S^1} e^\rho \, d\theta , \quad E:= \int_{S^1}e^{\rho -2\tau} J \, d\theta , \quad \text{and} \quad Y := \int_{S^1}  e^{l + \rho + 2\tau} \, d\theta.
    \end{align}
Note that in a sense \(E\) acts as a weighted Sobolev norm for the functions \(V,Q\).
    Unfortunately, \(E\) does not satisfy a good differential inequality for use with Gr\"onwall's inequality.
    It is thus necessary to modify \(E\) by a correction term \(\Lambda\), and to normalise it by a factor of \(\Pi\):
    \begin{align} \label{h def}
        H:= \Pi ( E + \Lambda).
    \end{align}
    We use $\Lambda$ to denote generic correction terms. The explicit form of $\Lambda$ may differ from case to case. For example, in the specific case discussed below, Eq. \eqref{cor def} gives the explicit form of $\Lambda$. The introduction of the normalisation introduces error terms which one hopes to bound by assumption.
    \item
    One may then form an ordinary differential equation system for $\Pi$, $H$, and $Y$, with error terms, make the ansatz that the error terms are small, and find fixed points for the resulting ODE system.
    \item
    To show that \( \Pi , H , Y\) actually converge to solutions of the ODE, we perform a bootstrap argument.
    If \( \Pi , H , Y\) are close enough to the values of the sink in the ODE, one hopes to show that the errors introduced by the correction and by the ODE are small, showing that \( \Pi , H , Y\) flow toward these sink values for the ODE system.
    \item
    One then recovers bounds on \(E\), which allows us to then prove estimates on \(V,Q\).
\end{enumerate}

We hope to use what we have learned from carrying out numerical simulations for $B \ne 0$ solutions to assess which portions of the proof for \(B=0\) solutions might be adapted for the more general case. 

In order to extend these or alternative methods of proof to the $B \ne 0$ case, it is important to understand at what point in, e.g., \cite{MR4062458}, setting $B = 0$ becomes necessary for further analysis. As described above, the method of proof in \cite{MR4062458} relies on developing ``correction terms’’ to convert one or more expressions to an ``energy.’’ If the correction terms can be bounded in the regime of interest, the behavior of energy terms in the same regime, e.g. asymptotically in time, may then be controlled. This is a standard method of proof.  In \cite{MR4062458}, we arrive at Eq. (3.4) in that paper as a consequence of the development of a necessary correction term which we rewrite here (using $\langle Q \rangle$, etc., to denote the spatial average of $Q$, etc.):
\begin{equation}
\langle e^\rho V_\tau \rangle = A + B \langle Q \rangle  + \int_{S^1}\, e^{2(V-\tau)} Q_\tau (Q - \langle Q\rangle ) e^\rho \, d\theta.
\label{frompaper1}
\end{equation}

Previous to this point in \cite{MR4062458}, it is shown that the last term on the right hand side of \eqref{frompaper1} is bounded by a constant whether or not $B = 0$. However, $B \langle Q \rangle $, a term arising from the introduction of the correction term needed for the subsequent proof, cannot be so bounded. In fact, numerical studies (as shown in Fig. \ref{proof}) show that the term grows as $e^{a \tau}$ and thus comes to dominate the asymptotics of $\langle e^\rho V_\tau \rangle$. Setting $B = 0$ (i.e., restricting to that subfamily), yields that $\langle e^\rho V_\tau \rangle$ is bounded by a constant allowing the proof to proceed. It may be feasible to develop proofs for the generic case where one expects asymptotically that $\langle e^\rho V_\tau \rangle \approx B\langle Q\rangle$ via a new model ODE framework to analyze the asymptotic dynamics.\footnote{The model ODE system used in \cite{MR4062458} may be found implicitly in Section 4 of that paper.}
\begin{figure}
    \centering
    \includegraphics[width=\textwidth]{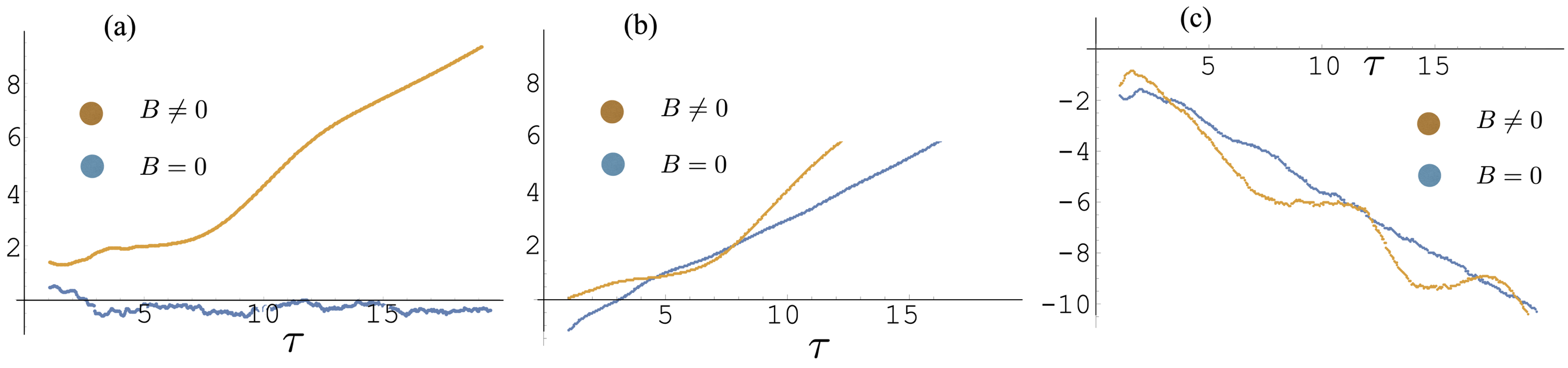}
    \caption{The individual terms in \eqref{frompaper1} for a representative simulation for the $B=0$ (blue) and $B \ne 0$ (orange) classes. (a) The entire right hand side. The term $\langle e^\rho V_\tau \rangle$ is very close to constant (taking on the value of $A$) for the $B=0$ case, while it shows growth for the $B\ne 0$ simulation. (b) The term $\langle Q\rangle$ is displayed and grows for both cases. The growth of $\langle e^\rho V_\tau \rangle$ is diminished only if $B=0$. (c) The remaining term on the right hand side of \eqref{frompaper1} decays in both cases and thus does not contribute to the failure of the entire term to be bounded if $B \ne 0$. }\label{proof}
\end{figure}

We note that the most useful form of the ODE system for this analysis involves the quantities
\mnote{BEVERLY?: Give the normalized \(\Pi, Y\) the names they have in the paper}
\begin{align}
    \frac{\Pi}{e^\tau \sqrt{H}}  \quad \text{and}\quad \frac{Y}{e^{3\tau} \sqrt{H}} \ .
\end{align}

These quantities are defined in Eqs. \eqref{pi def} and \eqref{h def} with Eq. \eqref{cor def} taken as the definition of $\Lambda$. We emphasize that the correction term $\Lambda$ takes on the value that is needed for a particular analysis. In this paper, we use 
\begin{align}\label{cor def}
    \Lambda :=& \frac{1}{2} e^{-2\tau} \int_{S^1} V_\tau  \left( V - \langle V \rangle - 1\right)e^\rho \, d\theta
\end{align}
as the necessary term.
\begin{figure}
    \centering
    \includegraphics[scale=.25]{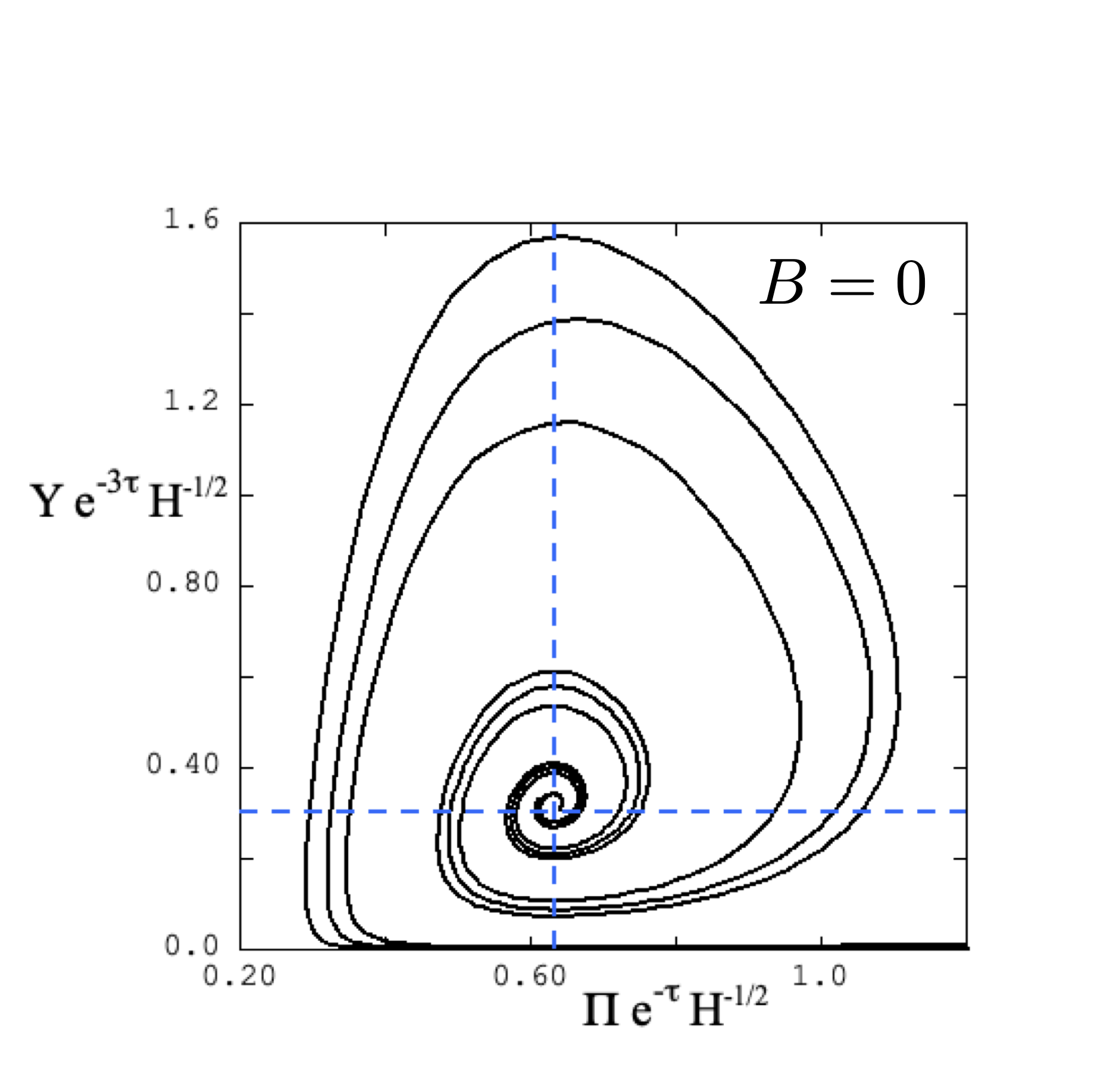}
    \caption{Plot of trajectories for three different $B=0$ simulations in the configuration space of \(Y e^{-3\tau} H^{-1/2}\) vs \(\Pi e^{-\tau} H^{-1/2}\) for three simulations in the $B=0$ class. The trajectories enter the displayed region from the lower right hand corner and spiral into an attractor indicated by the blue, dashed lines. The value of the attractor is $\left ( \frac{2}{\sqrt{10}},\frac{1}{\sqrt{10}} \right) \approx (0.63,\,0.31)$ where the first number is the component along the horizontal axis. The variables in this plot are obtained from Eqs. \eqref{pi def}, \eqref{h def}, and \eqref{cor def}.}\label{b0sink}
\end{figure}

\begin{figure}
    \centering
    \includegraphics[scale=.3]{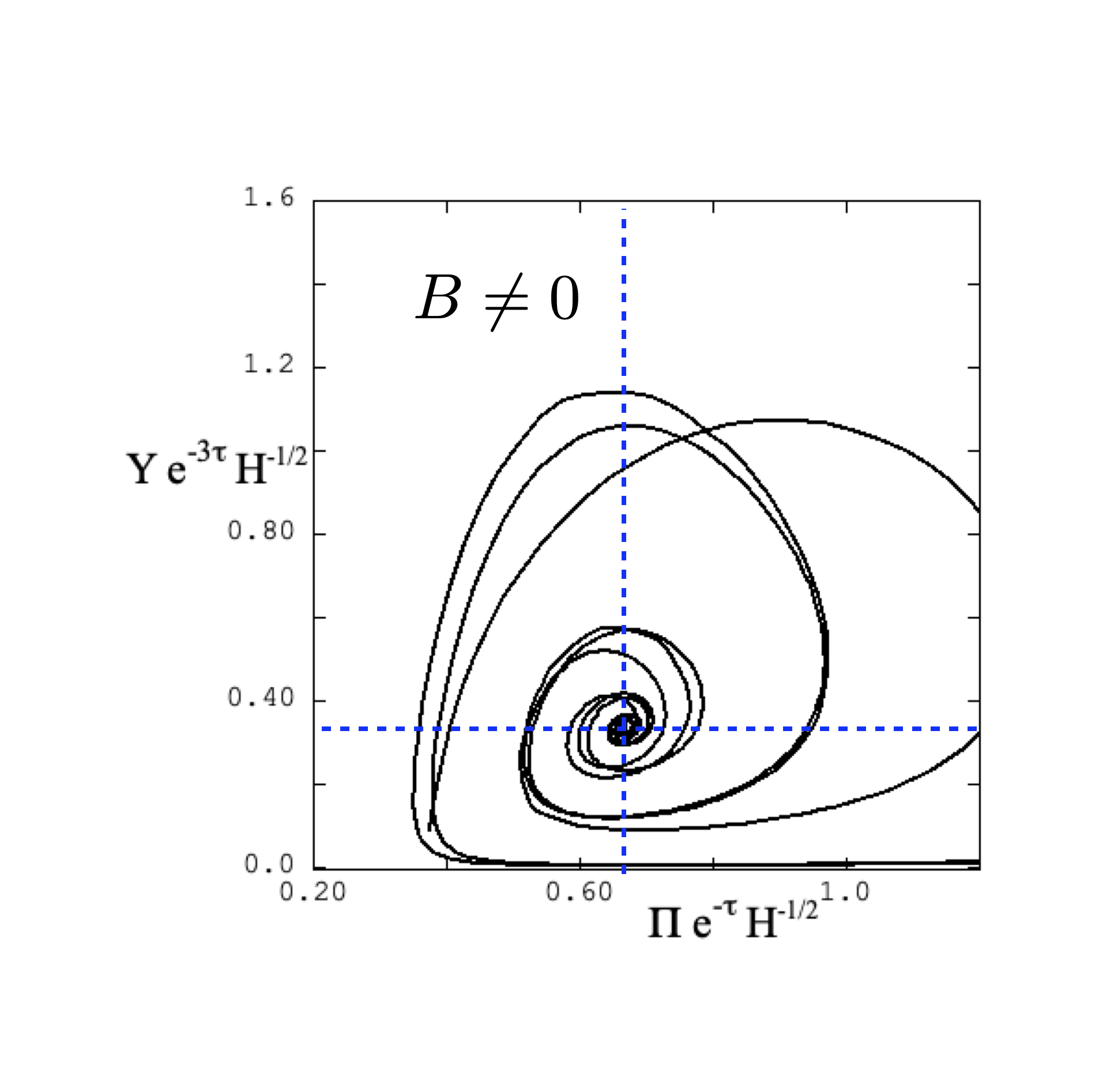}
    \caption{Plot of trajectories for three different $B\ne 0$ simulations in the configuration space of \(Y e^{-3\tau} H^{-1/2}\) vs \(\Pi e^{-\tau} H^{-1/2}\) for three simulations in the $B \ne 0$ class. The trajectories enter the displayed region from the lower right hand corner and spiral into an attractor indicated by the blue, dashed lines. The approximate value of the attractor is $(0.67,\,0.33)$ where the first number is component along the horizontal axis. The variables in this plot are obtained from Eqs. \eqref{pi def}, \eqref{h def}, and \eqref{cor def}.}\label{gensink}
\end{figure}
In Figs. \ref{b0sink} and \ref{gensink}, we demonstrate the existence of an attractor (in the sense meant in this paper) for the evolution in the configuration space defined by the spatially averaged functions \(Y e^{-3\tau} H^{-1/2}\) vs \(\Pi e^{-\tau} H^{-1/2}\); Fig. \ref{b0sink} does this for the $B=0$ case and Fig. \ref{gensink} for the $B\ne 0$ case. These trajectories are obtained using a spatial resolution of 2048. In Fig. \ref{sinkcvgc}, we indicate these trajectories in that same configuration space for two different spatial resolutions separately for the $B =0$ and $B \ne 0$ classes. As the figure shows, in both cases, the attractor appears to be independent of the chosen spatial resolution.  While the attractors for the $B=0$ and the $B\ne 0$ cases appear to be nearly the same as shown in Figs. \ref{b0sink} and \ref{gensink}, with a finer scaling, we show in Figure \ref{sinkdiff} that these attractors are different for the two cases.  Interestingly, our simulations indicate that the attractor for the $B \ne 0$ cases are independent of the specific value of $B$ although this may warrant further investigation.  
\begin{figure}
    \centering
    \includegraphics[scale=.35]{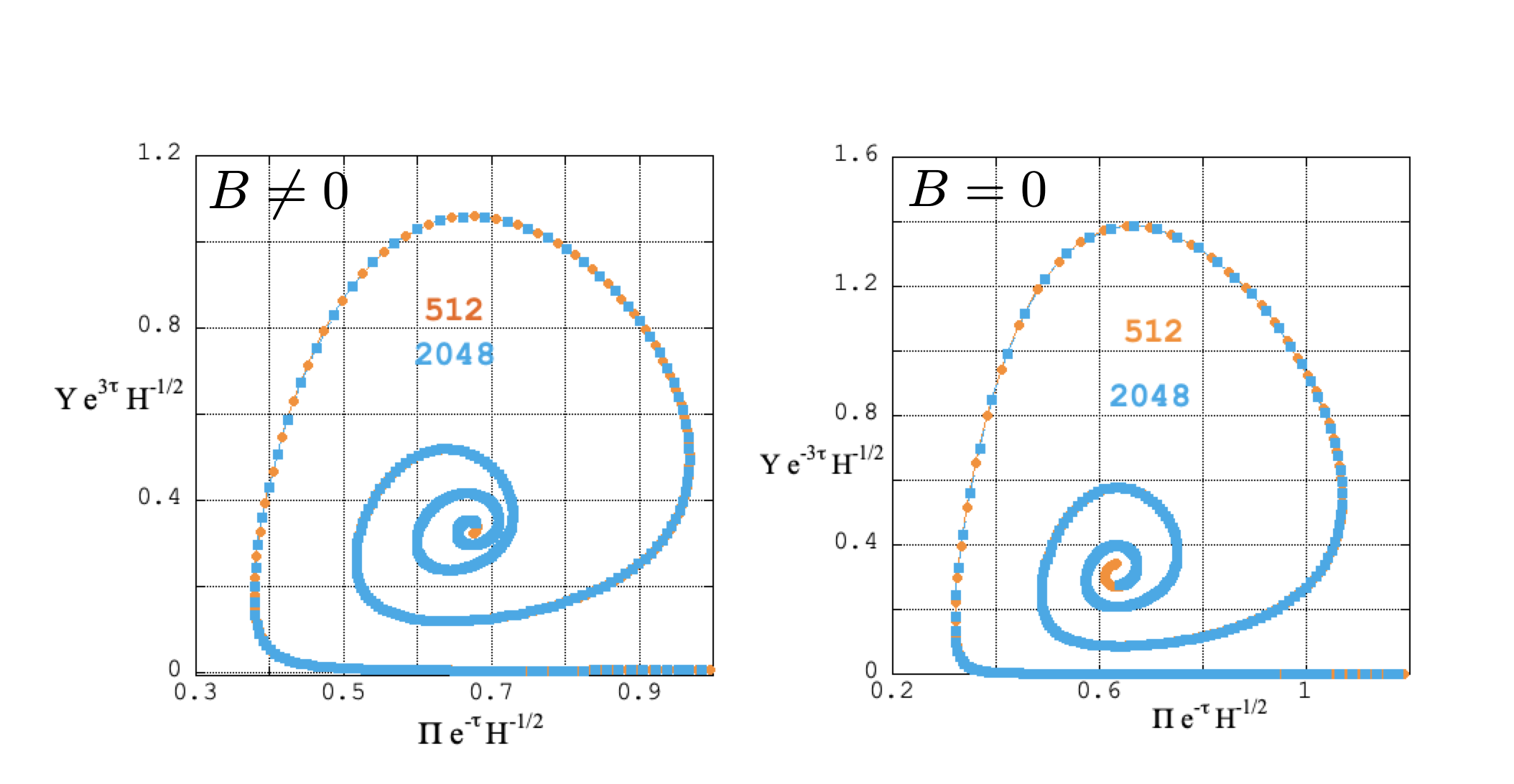}
    \caption{Plot of two trajectories from simulations in the $B\ne 0$ class (left) and $B = 0$ class (right) for spatial resolutions of 512 (orange) and 2048 (blue). The location of the attractor appears to be the same within a class and different between classes although higher precision analysis might reveal slight differences within a class due to spatial resolution differences.}\label{sinkcvgc}
\end{figure}
\begin{figure}
    \centering
    \includegraphics[scale=.25]{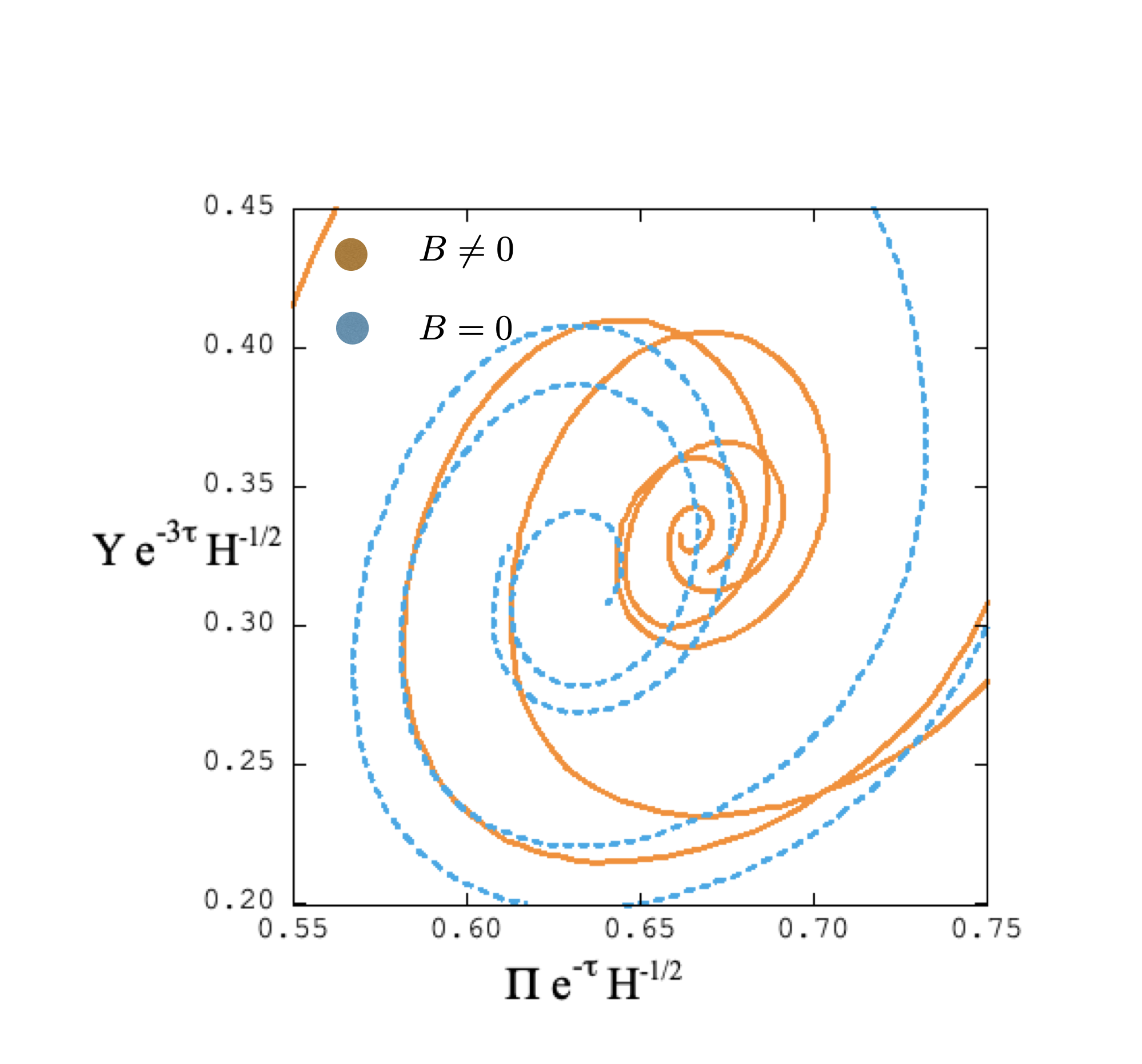}
    \caption{Plot of two trajectories from each of Fig. \ref{b0sink} -- blue dashed lines -- and Fig. \ref{gensink} --orange solid lines. We see in this figure that the attractors for the $B=0$ and the $B\ne 0$ cases are not the same.}\label{sinkdiff}
\end{figure}
\subsection{		Energy sloshing \label{sloshing}}
As in \cite{MR4062458}, we define ``energy'' variables $E$, $E_V$, and $E_Q$. The term ``energy'' reflects the field-energy-like expressions for \(E\), $E_V$ and $E_Q$ built from the gravitational wave polarizations $V$ and $Q$. Specifically, we set 
\begin{equation}
   E := E_V + E_Q 
    \label{Etot}
\end{equation}
where
\begin{equation}
E_V := \frac{1}{2} \int_{S^1}\, e^{\rho - 2 \tau} \left(V_\tau^2 + e^{-2\rho + 2\tau} V_\theta^2 \right) d\theta 
\label{ev}
\end{equation}
and
\begin{equation}
    E_Q := \frac{1}{2} \int_{S^1}\, e^{\rho - 2 \tau} \left(e^{2V-2\tau}Q_\tau^2 + e^{2V-2\rho } Q_\theta^2 \right) d \theta. \label{eq}
\end{equation}

\begin{figure}
    \centering
    \includegraphics[width=\textwidth]{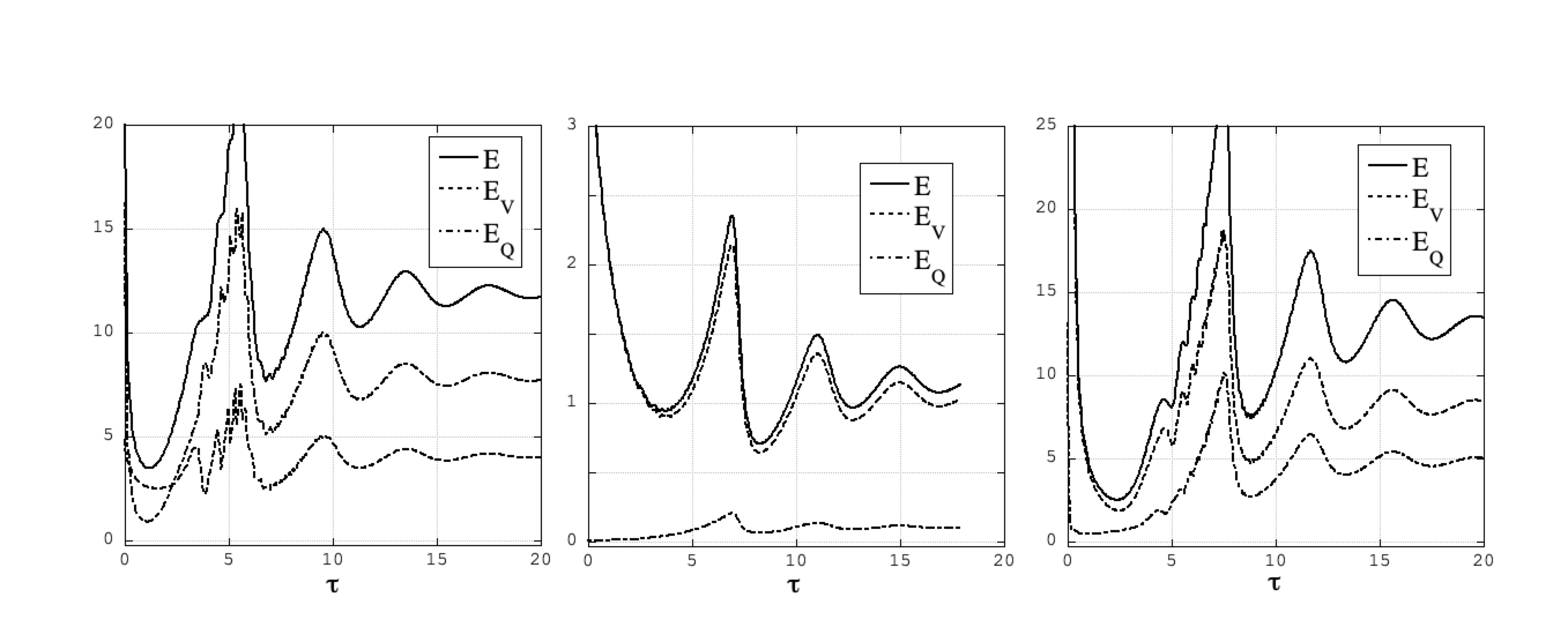}
    \caption{Three examples from simulations in the $B=0$ class showing the time evolution of $E,\ E_V, \ E_Q$ vs $\tau$. Note that in every case, the time dependence of $E_V$ and $E_Q$ follows that of $E$.}\label{sloshb0}
\end{figure}
\begin{figure}
    \centering
    \includegraphics[width=\textwidth]{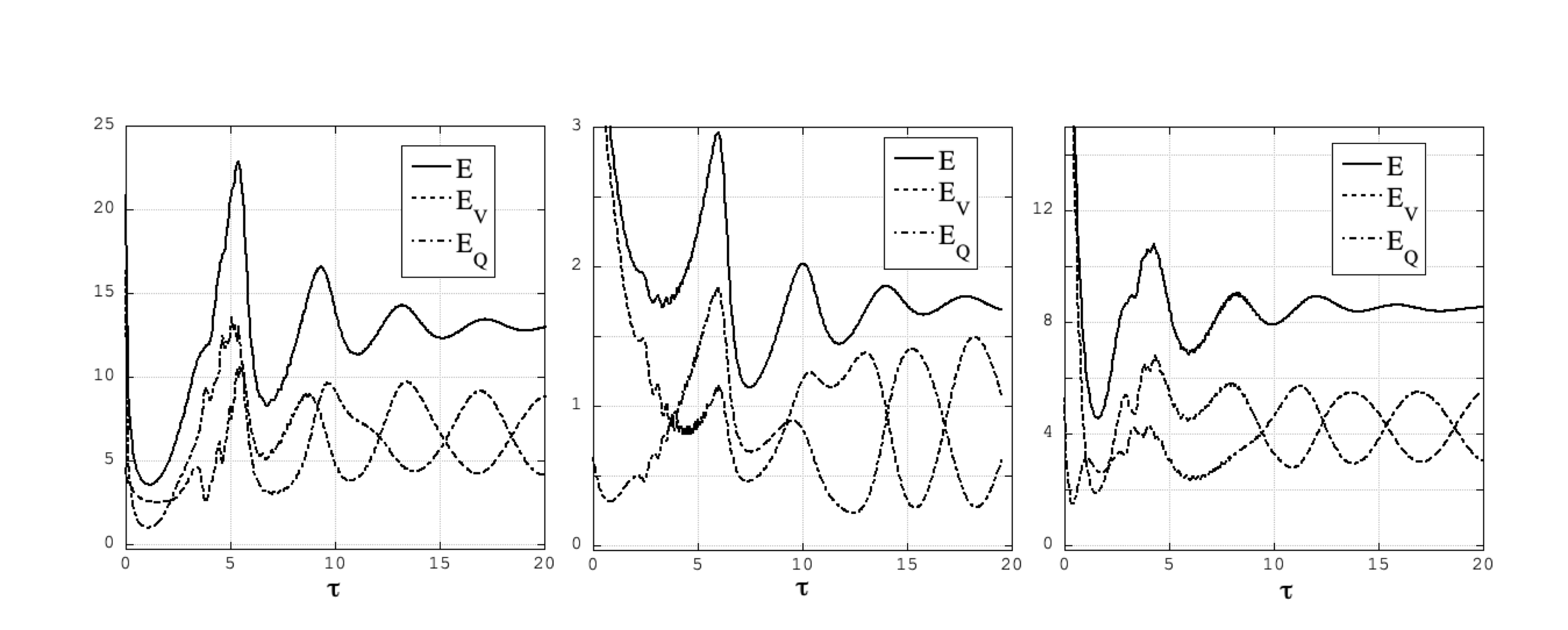}
    \caption{Three examples from simulations in the $B \ne 0$ class showing the time evolution of $E,\ E_V, \ E_Q$ vs $\tau$. Note that in every case, the time dependence of $E_V$ and $E_Q$ does not follow that of $E$ but rather $E_V$ and $E_Q$ appear to interact with each other -- an effect we call ``sloshing''.}\label{sloshgen}
\end{figure}
In Figs. \ref{sloshb0} and \ref{sloshgen}, we plot the evolution of these energy quantities for $B=0$ and $B\ne 0$ simulations. Comparing the behavior seen in these two figures, we see that the asymptotic behavior differs significantly for the two different cases. Very interestingly, we see in Fig. \ref{sloshgen} that the energy “sloshes” back and forth between the two gravitational polarizations in a manner reminiscent of the “equipartition" behavior of a coupled pendulum system for the $B\ne 0$ case, while for the $B=0$ case this does not occur.
 

Here we discuss a heuristic understanding of the “sloshing” phenomena seen in the $B\ne 0$ simulations.  The definition of $B$ in Eq. \eqref{b def} implies that the spatial average of $\pi_Q$, the canonically conjugate momentum to $Q$, vanishes if $B=0$. In some sense, since $\pi_Q = e^{2(V-\tau)} \frac{dQ}{ d\tau}$, some aspect of the $V$-$Q$ coupling is missing in the $B=0$ subclass but is present in the generic case. Our simulations indicate that the full coupling of the generic case is required to cause the sloshing (i.e., energy transfer) between the degrees of freedom indicated by the average energies.

Note that the restriction to the $B=0$ subfamily is necessary for carrying out the proofs in \cite{MR4062458}. The growth of the terms with $B\ne 0$ coefficients cannot be controlled in the model ODE system we consider there. 
It is possible that developing an alternative model ODE system which features dominant $B\ne 0$ terms could lead to a mathematical result which determines the asymptotic behavior of the generic $T^2$-Symmetric solutions in the expanding direction, at least assuming certain smallness restrictions.

At this time we have no rigorous explanation for the difference in behavior for the $B = 0$ and $B \ne 0$ ``energy'' asymptotics. As a possible consideration, we have noted above that the former class at some level removes part of the $Q$ degree of freedom. It may be possible therefore to show that this forces $Q$ to follow $P$ in phase. In the generic case, the two degrees of freedom are complete in the sense that no {\it a priori} restrictions are imposed on the $Q$-momentum (such as setting its spatial average to zero via $B=0$). It therefore may be possible to show that the nonlinear coupling between the two degrees of freedom leads to an equipartition of the energy. Methods used to prove our theorems for the $B = 0$ case fail in the generic case because the behavior of the correction term $\Lambda$ cannot be controlled. It grows in contrast to the other terms and thus defines a different effective system of approximate equations 
This will be a topic for future work.



\section{	Conclusions}\label{conclusions}

In this paper, we have provided detailed numerical evidence that $B \ne 0$ defines a class of $T^2$-Symmetric spacetimes with asymptotics different from the $B = 0$ case where rigorous mathematical results exist. The simulation code is shown to converge as expected with increasingly fine spatial resolution. In particular, the momentum constraint is shown to remain acceptably small. We have demonstrated that it is possible to find numerically the correct power laws for the asymptotic expansion of the spacetimes for generic values of $B$ and that, as shown in Table \ref{tab}, these power laws are different for the $B \ne 0$ and $B = 0$ cases. An example is shown in Fig. \ref{proof} of how, in contrast to the $B =0$ case, the $B \ne 0$ asymptotics obstruct the method of proof used in \cite{MR4062458} but also may suggest an alternative method. Finally, we show that the $B \ne 0$ spacetimes exhibit an oscillatory behavior in the energy-like terms associated respectively with the $V$ and $Q$ gravitational-wave polarizations. It is possible that this ``sloshing’’ behavior arises generically but not for the $B = 0$ subclass because a crucial part of the $Q$-degree of freedom is missing in the latter case.

While the $T^2$-Symmetric spacetime solutions of Einstein’s equations are not expected to be physically realistic models of the cosmos, the study of these solutions serves as an effective laboratory for investigating what sort of behavior one might expect to see in a larger class of cosmological models governed by Einstein’s gravitational field theory. Our focus in this paper is the behavior of the $T^2$-Symmetric spacetimes in the expanding direction rather than towards the singularity. Mathematical analyses in \cite{MR3513138} and \cite{MR4062458} have shown that subfamilies of the $T^2$-Symmetric solutions have characteristic “attractor” behavior in the expanding direction. While we have not been able to prove mathematically that the same behavior is found in the general class of $T^2$-Symmetric solutions, our numerical simulations of a very wide selection of these solutions in this paper support the contention that attractors do characterize the behavior of this general class in the expanding direction. Moreover, the attractors for the general class do not coincide with the attractors for the subfamilies of solutions, which indicates a degree of instability of the behavior of the subfamilies. Our simulations also indicate that the general class of solutions exhibits a “sloshing” behavior between the two gravitational polarizations, much like the equipartition behavior found in coupled pendulum systems. We present a very rough outline of a path for proving that the behavior seen in these numerical simulations holds for the general class of $T^2$-Symmetric spacetime solutions.

\begin{acknowledgments}
JAI received partial support for this work from NSF grants PHY-1306441 and PHY-1707427 to the University of Oregon. ANL was supported by PHY-1306441 to the University of Oregon during part of this work. 
\end{acknowledgments}

\bibliography{bib}                              

\end{document}